\documentclass[letterpaper,twocolumn,10pt]{article}
\usepackage{usenix-2020-09}
\usepackage[utf8]{inputenc}
\usepackage[lambda ,advantage, operators, sets , adversary, landau, probability, notions, logic, ff , mm, primitives, events, complexity, asymptotics, keys]{cryptocode}

\usepackage{pgfplots}
\pgfplotsset{compat=1.7}
\usepackage{url}
\usepackage{xurl}
\usepackage{amsmath}
\usepackage{amssymb}
\usepackage[utf8]{inputenc}
\usepackage[english]{babel}
\usepackage{amsthm}
\usepackage{stfloats}
\usepackage{circuitikz}
\usepackage{multirow}
\usepackage{booktabs}
\usepackage{subfig}
\usepackage{tikz}
\usepackage{xcolor}
\usepackage{algorithm}
\usepackage[noend]{algpseudocode}
\usepackage[numbers,sort&compress]{natbib}
\usepackage{thm-restate}
\usepackage{authblk}

\usepackage{cleveref}

\newcommand{\clientset}[0]{\mathbb{X}}
\newcommand{\clientsetsize}[0]{m}
\newcommand{\serverset}[0]{\mathbb{Y}}
\newcommand{\serversetsize}[0]{n}
\newcommand{\maxbinclient}[0]{\gamma}
\newcommand{\maxbinserver}[0]{\mu}

\newcommand{\hw}[0]{h}

\newcommand{\numbins}[0]{b}

\newcommand{\effectivebitlength}[0]{\bar{\lambda}}
\newcommand{\ex}[0]{$\times$}

\newcommand{\formatcomment}[1]{\scriptsize\textcolor{blue!25!black}{\texttt{#1}}}

\algrenewcommand{\algorithmiccomment}[1]{\hfill\parbox{5.2cm}{\formatcomment{//\,#1}}}
\algrenewcommand\algorithmicprocedure{\textbf{algorithm}}
\algnewcommand\algorithmicforeach{\textbf{for each}}
\algdef{S}[FOR]{ForEach}[1]{\algorithmicforeach\ #1\ \algorithmicdo}

\newtheorem{theorem}{Theorem}

\definecolor{redbg}{RGB}{254,241,240}
\definecolor{redoutline}{RGB}{252,163,152}
\definecolor{redtext}{RGB}{207,24,34}
\definecolor{beigebg}{RGB}{245, 245, 220} %

\definecolor{pepsired}{RGB}{201,0,44}
\definecolor{pepsiblue}{RGB}{39,81,184}
\definecolor{pepsipurple}{RGB}{120,41,114}

\newcommand{\cpsi}[0]{{\color{pepsiblue}PE}{\color{pepsired}PSI}}

\newtheorem{corollary}{Corollary}[theorem]

\newcommand{\BMB}[0]{(1024*1024*8)}

\begin{document}
%-------------------------------------------------------------------------------

%don't want date printed
\date{}

% make title bold and 14 pt font (Latex default is non-bold, 16 pt)
\title{\Large \bf \cpsi{}: Practically Efficient Private Set Intersection in the Unbalanced Setting}

%for single author (just remove % characters)
% \author{\rm Anonymous Author(s)} 

%for single author (just remove % characters)
% \author{
% {\rm Rasoul Akhavan Mahdavi}\\
% University of Waterloo
% \and
% {\rm Second Name}\\
% Second Institution
% copy the following lines to add more authors
% \and
% {\rm Name}\\
% Name Institution
% } % end author

% \author{
%     {\rm Rasoul Akhavan Mahdavi$^1$} \and
%     {\rm Nils Lukas$^1$} \and
%     {\rm Faezeh Ebrahimianghazani$^1$} \and
%     {\rm Thomas Humphries$^1$} \and
%     {\rm Bailey Kacsmar$^2$} \and
%     {\rm John Premkumar$^1$} \and
%     {\rm Xinda Li$^1$} \and
%     {\rm Simon Oya$^3$} \and
%     {\rm Ehsan Amjadian$^1$} \and
%     {\rm Florian Kerschbaum$^1$} \\
%     $^1$ University of Waterloo ~ $^2$ University of Alberta ~ $^3$ University of British Columbia 
    % {
    %     University of Waterloo \and University of Alberta \and University of British Columbia
    % }
% }

\author[1]{Rasoul Akhavan Mahdavi}
\author[1]{Nils Lukas}
\author[1]{Faezeh Ebrahimianghazani}
\author[1]{Thomas Humphries}
\author[2]{Bailey Kacsmar}
\author[1]{John Premkumar}
\author[1]{Xinda Li}
\author[3]{Simon Oya}
\author[1, 4]{Ehsan Amjadian}
\author[1]{Florian Kerschbaum}
\affil[1]{University of Waterloo}
\affil[2]{University of Alberta}
\affil[3]{University of British Columbia}
\affil[4]{Royal Bank of Canada}
\affil[ ]{\textit {\{rasoul.akhavan.mahdavi, nlukas, f5ebrahi, thomas.humphries, jpremkumar, xinda.li, ehsan.amjadian, florian.kerschbaum\}@uwaterloo.ca, kacsmar@ualberta.ca, simon.oya@ubc.ca} }

% make the title area
\maketitle

\begin{abstract}
    Two parties with private data sets can find shared elements using a Private Set Intersection (PSI) protocol without revealing any information beyond the intersection.
    Circuit PSI protocols privately compute an arbitrary function of the intersection - such as its cardinality, and are often employed in an \emph{unbalanced} setting where one party has more data than the other. 
    Existing protocols are either computationally inefficient or require extensive server-client communication on the order of the larger set. 
    We introduce Practically Efficient PSI or \cpsi{}, a non-interactive solution where only the client sends its encrypted data. 
    \cpsi{} can process an intersection of 1024 client items with a million server items in under a second, using less than 5 MB of communication. 
    Our work is over 4 orders of magnitude faster than an existing non-interactive circuit PSI protocol and requires only 10\% of the communication.
    It is also up to 20 times faster than the work of Ion et al., which computes a limited set of functions and has communication costs proportional to the larger set.
    Our work is the first to demonstrate that non-interactive circuit PSI can be practically applied in an unbalanced setting.
\end{abstract}

% Two parties with private data sets can find shared elements using a Private Set Intersection (PSI) protocol without revealing any information beyond the intersection. Circuit PSI protocols privately compute an arbitrary function of the intersection - such as its cardinality, and are often employed in an unbalanced setting where one party has more data than the other. Existing protocols are either computationally inefficient or require extensive server-client communication on the order of the larger set. We introduce Practically Efficient PSI or PEPSI, a non-interactive solution where only the client sends its encrypted data. PEPSI can process an intersection of 1024 client items with a million server items in under a second, using less than 5 MB of communication. Our work is over 4 orders of magnitude faster than an existing non-interactive circuit PSI protocol and requires only 10% of the communication. It is also up to 20 times faster than the work of Ion et al., which computes a limited set of functions and has communication costs proportional to the larger set. Our work is the first to demonstrate that non-interactive circuit PSI can be practically applied in an unbalanced setting.

\section{Introduction}

\begin{table*}[t]
    \centering
    \caption{
        Comparing the properties of different PSI protocols.
        $m$ and $n$ are the client and server set sizes, respectively. We assume $m\ll n$ in the analysis of works which support the unbalanced setting.
        DH: Diffie-Hellman. 2PC: two-party computation. PDTE: private decision tree evaluation. *PSI using \cite{ion2020deploying, yingPSIStatsPrivateSet2022, maSecureComputationFriendlyPrivateSet2022a} is done in one round and an extra round is required to compute the sum. Moreover, it can only perform a limited set of functions in only one extra round.
        **PSI is realized by extending private set membership. 
    }
    \label{tab:compare-protocols}
    \begin{tabular}{c|c|c|c|c|c|c|c|c}
    \toprule
        Work & Tools & {\small Unbalanced} & Rounds & Comm. & Comp. & PSI & Labelled PSI & Circuit PSI \\
    \midrule
        \cite{pinkas2018scalable} & B-OPPRF + 2PC & \ex & $1 + \log_2 \lambda$ & $O(n)$ & $O(n\log n)$ & \checkmark & \checkmark & 2PC \\
        \cite{chandran2022circuit-psi} & RB-OPPRF + 2PC & \ex & $1 + \log_2 \lambda$ & $O(n)$ & $O(n)$ & \checkmark & \checkmark & 2PC \\
    \midrule
        \cite{ion2020deploying, yingPSIStatsPrivateSet2022} & DH & \ex & 1+1 $^{*}$ & $O(n)$ & $O(n)$ & \checkmark & \ex & \checkmark$^{*}$ \\
        \cite{maSecureComputationFriendlyPrivateSet2022a,ciampiCombiningPrivateSetIntersection2018} & OT + PDTE & \ex & 1+1 & $O(n)$ & $O(n)$ & \checkmark & \checkmark$^{**}$ & \checkmark$^{**}$ \\
        \cite{chen2017fast} & HE & \checkmark & 1 & $O(m\log n)$ & $O(n\log n)$ & \checkmark & \checkmark & \ex \\
        \cite{chen2018LabelledPSI, cong2021LabelledPSI} & OPRF + HE & \checkmark & 2 & $O(m)$ & $O(n\log n)$ & \checkmark & \checkmark & \ex \\
    \midrule
        DiPSI~\cite{kacsmar2020differentially} & HE & \checkmark & 1 & $O(m)$ & $O(n)$ & \checkmark & \checkmark & \checkmark \\
        \cpsi{} & HE & \checkmark & 1 & $O(m)$ & $O(n)$ & \checkmark & \checkmark & \checkmark \\
    \bottomrule
    \end{tabular}
\end{table*}  
Privacy-preserving data analytics operates on the principles that (i) data is accessible and can be analyzed while (ii) ensuring that the user's privacy remains uncompromised. 
For instance, consider the application of secure contact discovery used in private messengers such as Signal~\cite{pirpsi, kales2019mobile}. 
Users want to connect with their friends without sharing their entire contact list with the server, and the server does not want to disclose the list of all contacts to each user. 
The revelation of this information by the client (or server) would defeat the very essence of these services.

Private Set Intersection (PSI) protocols offer a solution to this problem.
PSI is designed to compute the intersection of two sets while ensuring that nothing beyond the intersection is revealed. 
A specific form of PSI is \textit{unbalanced} PSI, in which one party, typically the client, has a substantially smaller set than the second party, the server.
Moreover, the client is often a resource-constrained device like a mobile phone or an IoT device with limited network connectivity.
Consequently, it benefits the client by minimizing network communication, reducing the number of network round trips, and offloading computational tasks to the server.

In some cases, the goal of the client and server is to compute functions of the intersection, rather than to learn the intersection itself. For instance, in the case of COVID-19 exposure applications, a user simply wants to know if they share any randomly generated codes with those stored on the server, indicating potential exposure~\cite{reichert2021circuit,takeshita2021provably, duong2020catalic, Trieu2020EpioneLC}; who those matches are is irrelevant.
Similarly, when gauging the effectiveness of ad campaigns by cross-referencing online advertisements with offline credit card transactions~\cite{ion2020deploying}, the advertiser is interested in knowing the total sales corresponding to a specific set of credit cards used at certain vendors.
Only the total sales amount is of interest, not the individual transactions.
\textit{Circuit PSI} is the term used when a function of the intersection is calculated rather than the intersection itself.

In the circuit PSI setting, the work of Kacsmar et al.~\cite{kacsmar2020differentially} (referred to as DiPSI) is best suited for the unbalanced non-interactive setting. However, DiPSI still suffers from inefficiency issues. 
In contrast to DiPSI, the numerous fast state-of-the-art solutions for PSI~\cite{cong2021LabelledPSI,chandran2022circuit-psi,pinkas2018efficient} are limited in terms of being extendable to the circuit PSI setting.
Specifically, each state-of-the-art solution is limited in one of the following ways.
The first limitation is for solutions that are restricted as to the functions they can compute~\cite{ion2020deploying, pjac}.
They are unable to execute functions beyond a specific predefined set.
Despite the fact that protocols designed for specific functions allow targeted optimizations, our evaluation shows that we outperform the state-of-the-art for many specific functions, particularly for large server set sizes.
Furthermore, targeted optimizations become obsolete when the function needs to be changed.
Second, adapting some solutions to circuit PSI levies a computational and communication burden on the client to compute arbitrary functions, usually in the form of two-party computation~\cite{chandran2022circuit-psi, pinkas2019efficientcircuitbased}.
This communication burden is usually proportional to the larger set~\cite{ion2020deploying,chandran2022circuit-psi, pinkas2019efficientcircuitbased} or depends on the complexity of the task.
Such a burden is not acceptable when designing a non-interactive protocol.

\paragraph{Summary of Contributions.}
In this work, we propose \cpsi{}, an efficient, non-interactive circuit PSI protocol using homomorphic encryption.
Homomorphic Encryption (HE) is a form of encryption that permits computation on the data while in encrypted form.
In \cpsi{}, the client encrypts its elements using an HE scheme and sends them to the server, which compares the elements homomorphically.
The output of each comparison is binary and is used to compute functions.
Our approach removes the limitations that are present in existing work.
\cpsi{} does not restrict the function that can be evaluated, there is no communication or computation burden on the client to compute the function, and it only has communication complexity proportional to the smaller client set size.

Our approach also overcomes the impractical runtimes associated with the DiPSI protocol from the literature~\cite{kacsmar2020differentially}.
We address the slow comparisons by comparing items using the efficient constant-weight equality operator~\cite{mahdavi2022constant}, permutation-based hashing, and cuckoo hashing. Altogether this reduces the size of the elements that need to be compared such that \cpsi{} is over four orders of magnitude faster than DiPSI and requires 90\% less communication.
For example, \cpsi{} can compute the intersection of 1024 and one million 32-bit elements in under one second with less than 5~MB of communication.

\cpsi{} has a runtime and communication cost comparable to other PSI protocols based on HE~\cite{chen2018LabelledPSI,cong2021LabelledPSI}, which cannot extend to circuit PSI.
Moreover, the client is required to do much less computational work in \cpsi{}, making it compatible with use cases where the client does not have much computational power, i.e., the unbalanced setting.
We empirically compare \cpsi{} with state-of-the-art protocols in our experiments, showing that \cpsi{} is even competitive with PSI protocols which cannot extend to circuit PSI.
Our work is the first to show that non-interactive circuit PSI can be efficient up to the point it is practical for use in applications.

\section{Related Work}

The term \textit{private set intersection} (PSI) was coined by Freedman et al.~\cite{psi04}.
In a PSI protocol, two parties compute the intersection of their respective private sets such that no information is revealed except the intersection itself.
If the intersection is only revealed to one of the parties, this is known as one-sided PSI.
However, one-sided PSI can be extended to two-sided PSI by running the protocol a second time and switching the roles of the client and server.

\subsection{PSI with Leakage}
In some applications and real-world use cases, such as checking for compromised credentials, the overhead of PSI is high if the client element is compared to all elements.
For example, comparing the client credentials to billions of leaked credentials is expensive.
To alleviate the overhead, some works propose to only compare the client elements to a subset of the server elements~\cite{li2019protocols, thomasProtectingAccountsCredential2019}.
For example, elements are distributed into buckets by their prefix. To query for a specific element, the client only queries the bucket corresponding to that element instead of the entire set.
The server will know that the client's query falls into that specific bucket so such an approach partially leaks information about the client at the cost of better performance.
This leakage degrades the security guarantees but may be acceptable in some applications. 
Thomas et al.~\cite{thomasProtectingAccountsCredential2019} predict that the prefix of a user's credentials is common with about 350k-470k other users, assuming each client has one credential in the database.

Our work does not follow this approach and instead offers PSI without leakage.
While such an approach requires more computational effort, it offers better security.
We note that our protocol could also be combined with bucketization, offering a tunable trade-off between security and performance from the leakage of current protocols~\cite{li2019protocols, thomasProtectingAccountsCredential2019} to no leakage.

\subsection{Unbalanced PSI}
Unbalanced PSI is a special case of PSI where one party, which we call the client, has a much smaller set compared to the other party, i.e., the server~\cite{kales2019mobile, chen2017fast, chen2018LabelledPSI, cong2021LabelledPSI, aranha2022laconic, alamati2021laconic, li2019protocols, duong2020catalic}.
This is a common assumption in many applications of PSI.
For example, in compromised credentials checking (C3)~\cite{li2019protocols, kales2019mobile}, the client wishes to access a large database containing credentials that are leaked on the web.
While the set of client credentials is small, e.g., a few hundred credentials, the database of leaked credentials is in the billions and growing~\cite{li2019protocols, kales2019mobile}.
Moreover, the client is usually bandwidth-constrained and has limited computational power.
Hence, protocols with limited network round trips and little client-side computation are preferred.
A narrower variant of unbalanced PSI exists, dubbed laconic PSI, with the additional constraint that only the server learns the output of the protocol~\cite{aranha2022laconic, alamati2021laconic}.

While there are many efficient state-of-the-art PSI protocols~\cite{pinkas2015phasing, pinkas2018efficient, pinkas2018scalable, chandran2022circuit-psi, maSecureComputationFriendlyPrivateSet2022a}, the communication complexity of these solutions scales with the larger set, making them unsuitable for unbalanced PSI.
For example, Pinkas et al. and Chandran et al. propose PSI protocols that compute an oblivious PRF for each element of the server set~\cite{pinkas2018scalable, chandran2022circuit-psi}. 
Ma and Chow propose to construct a private decision tree from the server set, which is encrypted and sent to the client.
The client then obliviously traverses this decision tree using OT~\cite{maSecureComputationFriendlyPrivateSet2022a}.
Such solutions are useful in high-bandwidth, low-latency networks but fall short over high-latency networks.

% Kayes?
Amongst solutions that have no leakage, the current most efficient non-interactive PSI protocols in the unbalanced case are based on homomorphic encryption~\cite{chen2017fast, chen2018LabelledPSI, cong2021LabelledPSI}.
Homomorphic encryption (HE) is a type of encryption that allows computations on data in its encrypted form.
Chen et al.~\cite{chen2017fast} were among the first to propose a PSI protocol specifically designed for the unbalanced setting using homomorphic encryption.
If $\serverset$ represents the server set, the key idea is to construct the polynomial $P(x)=r \prod_{y_i\in\serverset} (x-y_i)$, for some random value $r$. The client sends an encryption of its elements to the server, which then homomorphically evaluates $P(x)$ on the client input.
For a certain client element $x_0$, if $x_0\in\serverset$, then $P(x_0)=0$, otherwise $P(x_0)$ is a random number.
The idea of representing the set as a polynomial had been proposed before~\cite{10.1007/11535218_15} and Chen et al. included optimizations to reduce the multiplicative depth of the function that is homomorphically evaluated, thereby enhancing efficiency~\cite{chen2017fast}.
Chen et al.~\cite{chen2018LabelledPSI} and Cong et al.~\cite{cong2021LabelledPSI} built upon \cite{chen2017fast}, employing a combination of Oblivious Pseudorandom Functions (OPRF) and homomorphic encryption.
This improved performance extended the security to the malicious model and also allowed the protocol to extend to elements of arbitrary bitlength.

\subsection{Labelled PSI}
A specific variant of PSI is \textit{labelled PSI}, where the server has a private label associated with each element in its private set. 
The server stores pairs in the form of $(y_i, \ell_i)$ where $y_i$ is the identifier of the element and $\ell_i$ is the label associated with that element.
The two parties compute the intersection of their sets and output all the pairs $(y_i, \ell_i)$ where $y_i$ is in the intersection of the two sets.
Nothing beyond the elements in the intersection and their corresponding labels are revealed to either one of the parties.
When the client has only one element, this problem is equivalent to private information retrieval by keywords, first proposed by Chor et al.~\cite{chor1997private}.
In the context of private contact discovery, the client may wish to retrieve the public keys of users in their contact list~\cite{pirpsi, kales2019mobile}.

Many PSI protocols can be extended to labelled PSI with additional computation and communication costs.
For example, Chen et al.~\cite{chen2018LabelledPSI} and Cong et al.~\cite{cong2021LabelledPSI} extend their protocol to labelled PSI by evaluating another polynomial which evaluates to the $\ell_i$ when $x=y_i$.
For the protocols of Pinkas et al.~\cite{pinkas2019efficientcircuitbased} and Chandran et al.~\cite{chandran2022circuit-psi} the extension to labelled PSI can be done without changing the asymptotic complexity, but it does result in doubling the concrete communication and computation cost.

\subsection{Circuit PSI}

Another variant of PSI is \textit{circuit PSI}, where the objective is to compute a function of the intersection, rather than the intersection itself. One example is PSI-Cardinality, where the parties compute the size of the intersection~\cite{kacsmar2020differentially}. Another example is PSI-Sum~\cite{ion2020deploying}, where the sum of values associated with elements in the intersection is revealed.

One application for PSI-Sum, known as ad-conversion, is assessing the effectiveness of ad campaigns~\cite{ion2020deploying}. A company purchasing an ad through an ad service company, such as Google, would like to know the total purchases made by users who have seen their particular ad. If Google can show a client who paid to have an ad displayed through their system that the ad leads to sales, then it is easier to convince the client to purchase more ads in the future. 
Such ad-conversion computations can be done by linking those who have seen an ad with their credit card transactions.
However, credit card vendors are unwilling to disclose their clients' transactions. PSI-Sum is a solution in this scenario.

In the existing approaches~\cite{chen2017fast,chen2018LabelledPSI, cong2021LabelledPSI}, every comparison between a client and server element results in an arithmetic output. An arithmetic output means that if the elements are equal, the output is zero, and the output is a random non-zero number in $\ZZ_p$, for some $p$.
While using this approach results in fast PSI protocols, computing a function of the intersection is not feasible.
To compute a function, the result of each comparison must be converted to a binary output, i.e., one for a match and zero otherwise.

One approach to computing such functions is to use secure two-party computation (2PC) to convert the arithmetic output into a secret-shared binary output.
2PC can then be used to compute any arbitrary function over the binary output.
While there is flexibility regarding the functions that can be computed in this approach, the communication complexity of MPC protocols is high.
More specifically, the communication complexity depends on the size of both sets and the complexity of the desired function~\cite{chandran2022circuit-psi, pinkas2019efficientcircuitbased}.

Another approach is to compare elements from the client and server homomorphically via a homomorphic equality circuit. 
The output of the comparisons is binary in this case.
This is the only non-interactive approach in the literature and is more suitable for the unbalanced setting since it does not require any client interaction.
The work of Kacsmar et al.~\cite{kacsmar2020differentially}, which they call DiPSI, uses this approach.
DiPSI achieves asymptotically optimal communication, given that only the client dataset is communicated over the network.
The authors use this protocol to compute a differentially private PSI-Sum and other functions.
The main problem with DiPSI is the extremely high runtime.
PSI-Sum between 1024 client elements and 1 million server elements can take over 600 minutes.
This can mainly be attributed to the homomorphic comparison function that they employ.
We propose a solution that compares elements homomorphically, similar to DiPSI but is over four orders of magnitude faster.

Some works take a different approach to circuit PSI by designing a protocol that can compute a very specific function very efficiently~\cite{pjac, ion2020deploying}.
This is in contrast to the other approaches where an arbitrary function can be derived.
A specialized protocol outperforms generalized protocols in terms of performance but restricts the functions that can be derived and may leak additional information in the process~\cite{ion2020deploying}.

Ion et al.~\cite{ion2020deploying} propose a protocol for PSI-Sum-with-Cardinality.
This protocol outputs, to the server, the sum of values associated with elements in the intersection, whilst also revealing the intersection cardinality to the client.
Lepoint et al.~\cite{pjac} proposed Private Join and Compute (PJC) which provides PSI-sum, which is the sum of values associated with elements in the intersection, without leakage of the intersection cardinality.
However, their protocol cannot compute any function other than the sum and has an expensive offline phase.
PSI-Stats follows a similar approach and supports more functions including arithmetic and geometric mean, standard deviation, and approximate composition~\cite{yingPSIStatsPrivateSet2022}.
However, this protocol is still limited to the functionalities proposed by the authors.
A common feature of these works is they require three rounds of interaction and the final results are revealed to the server.

\section{Background}\label{sec:background}

\subsection{The FV Cryptosystem}
The FV cryptosystem~\cite{fan2012somewhat} permits computation over vectors of numbers using Single Instruction, Multiple Data (SIMD) operators. Plaintexts are vectors of length $N$ where each element is in $\ZZ_t$. $N$ and $t$ are called the \textit{polynomial modulus degree} and the \textit{plaintext modulus}, respectively. Ciphertexts are vectors of polynomials with coefficients in $\ZZ_q$, where $q$ is called the \textit{ciphertext modulus} and $\mathcal{C}$ is the ciphertext space.
The following operations are permitted in FV:
% Addition ($+$) and multiplication ($\odot$) over vectors in $\ZZ_t$ are defined element-wise. 
\begin{itemize}\itemsep0mm
    \item Addition: Given $c_X,c_Y\in \mathcal{C}$ which encrypts $X,Y\in\ZZ_t^N$, respectively, output $c_A$ which encrypts $X+Y$.
    \item Plain Multiplication: Given $X\in\ZZ_t^N$ and $c_Y\in \mathcal{C}$ which encrypts $Y\in\ZZ_t^N$, output $c_{PM}$ which encrypts $X\odot Y$.
    \item Multiplication: Given $c_X,c_Y\in \mathcal{C}$ which encrypt $X,Y\in\ZZ_t^N$, respectively, output $c_M$ which encrypts $X\odot Y$.
\end{itemize}

Note the addition ($+$) and multiplication ($\odot$) operations from FV over vectors in $\ZZ_t$ are defined element-wise. 

\subsection{Constant-weight Encoding}
A constant-weight code is a collection of binary codewords with the same Hamming weight, i.e., each codeword has a fixed number of bits set to one. Throughout this paper, we denote the common Hamming weight that the codewords share as $h$.
We denote the set of $\ell$-bit constant-weight codewords with Hamming weight $h$ by $CW(\ell,h)$.

We use the arithmetic constant-weight equality operator from Mahdavi and Kerschbaum~\cite{mahdavi2022constant} for comparing constant-weight codewords with Hamming weight $\hw$.
The choice of $\hw$ impacts the computation cost of the operator and the required communication.
This operator does not depend on the operands which it is comparing, and so can be computed using SIMD operations.
\Cref{alg:arith-cw-eq} describes this operator. Throughout this paper, we denote $[n] = \{0,1,...,n-1\}$ for $n\in\NN$. 

\begin{algorithm}
	 \caption{Arithmetic Constant-weight Equality Operator}
	 \label{alg:arith-cw-eq}   
	 \begin{algorithmic}[1]
        \Procedure{Arith-CW-Eq}{$x, y$}
	 	\State $h' = \sum_{i\in[\ell]}x_i \cdot y_i$
	 	\vspace{1mm}
	 	\State $e = (1/h!) \prod_{i\in[h]} (h' - i)$
	 	\vspace{1mm}
            \State \textbf{return} $e \in \{0,1\}$
        \EndProcedure
	 \end{algorithmic}
\end{algorithm}

\subsection{Hashing Optimizations}\label{sec:binning-strategy}
We use a collection of hashing techniques to reduce the total number of comparisons and the cost of each comparison.

\emph{Hashing-to-Bins.}
Hashing-to-bins reduces the number of comparisons in a PSI protocol~\cite{pinkas2015phasing, pinkas2018scalable, freedman2016efficient, kacsmar2020differentially}.
Elements are distributed into a predetermined number of bins and comparisons are only performed between client and server elements in the same bin; since the bin in which an element is placed is determined using hashes of that element.
When using this strategy, note that it is necessary to pad each bin with dummy elements up to a predetermined size to prevent the load of the bins from leaking information about the dataset.  

There is a possibility that the number of elements in a bin exceeds the predetermined maximum, which would constitute a failure in the protocol. 
Thus, it is necessary to employ a strategy that cleverly manages elements that fall into the same bin. 

Cuckoo hashing~\cite{TCS-070, cuckoohashing2004} is a strategy for hashing elements into bins such that elements can be ejected from their current bin and placed in another bin after their initial placement.
In Cuckoo hashing~\cite{cuckoohashing2004}, there are $m$ elements which are placed in $b=\epsilon m$ bins, for some constant $\epsilon$.
Each bin holds at most one element and the placement of each element is determined using $k$ hash functions, denoted as $H_i: \mathcal{X}\mapsto [b]$ for $i=1,2,\cdots,k$.

For each element $x$, we place $x$ in bin $H_1(x)$. If bin $H_1(x)$ was already occupied by $x'$, such that $H_1(x)=H_{i}(x')$, we evict $x'$ from the bin. The new placement for $x'$ is bin $H_{i+1 \mod k}(x')$.
If that bin is empty, we are done; otherwise, we repeat this process, evicting the current resident of the bin.
To ensure the protocol terminates, it is necessary to define an upper bound on the number of evictions permitted.
If an empty bin is not found with the permitted number of evictions, the element is placed in a stash with a fixed size.
Thus, to locate an element $x$, it is sufficient to look in bin $H_1(x)$, $H_2(x)$, $\cdots$, $H_k(x)$ or in the stash.
When selecting parameters $\epsilon$, $k$, and the stash size, please refer to past work analyzing the failure rates of Cuckoo hashing~\cite{kirsch2010more, cuckoohashing2004, pinkas2018efficient, Arbitman2010BackyardCH}.

\emph{Permutation-based Hashing}
Permutation-based hashing (PBH) is a technique to reduce the length of elements~\cite{Arbitman2010BackyardCH, pinkas2015phasing,lam2016breaking,pinkas2018scalable}.
We describe PBH in the following paragraph using the assumption that the number of bins is $b=2^c$, for some $c\in \NN$. 

For a $\lambda$-bit element $x$, assume $x=x_H || x_L$ where $||$ is the concatenation operation and $|x_H|=c$ and $|x_L|=\lambda-c$. In the binning procedure, instead of placing $x$ in bin $H(x)$, we place $x_L$ in bin $b_x=x_H \xor H(x_L)$. Now, if two elements $x$ and $y$ are placed in the same bin and match, this means that $x_L = y_L$ and $b_x = b_y$. From that, we have
\begin{align*}
    b_x &= b_y \Rightarrow x_H \xor H(x_L) = y_H \xor H(y_L)  \Rightarrow x_H = y_H
\end{align*}

Combining $x_L = y_L$ and $x_H = y_H$, we can deduce that $x=y$.
Using this technique, we compare $\bar{\lambda}$-bit elements where $\bar{\lambda}=\lambda-c$.
Since we can identify matches using only these shorter elements, we will only store elements of length $\bar{\lambda}$ in our bins. Thus, we term $\bar{\lambda}$ the \textit{effective bitlength} and use it as the bitlength parameter throughout our descriptions and analysis of our protocol.

\begin{figure}[H]
    \centering
    \includegraphics[width=\columnwidth]{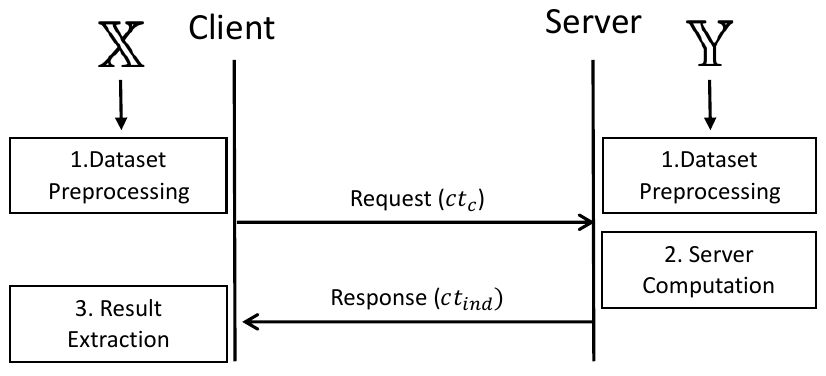}
    \caption{Stages of \cpsi{}.}
    \label{fig:overview}
\end{figure}

\section{\cpsi{}: Practically Efficient PSI} \label{sec:alg-desc}

In this section, we describe our protocol \cpsi{} in detail; specifically, we describe the necessary steps to compute the intersection between a client set of size $n$ and a server set of size $m$, where each set contains $\lambda$-bit elements.
\cpsi{} is a PSI protocol that operates in one round and consists of three stages: Dataset Preprocessing, Server Computation, and Result Extraction. The client sends a \textit{request message} to the server and the server responds with a \textit{response message}.
The request and response message constitute the total communication complexity of the protocol.
At the end of the protocol, the client learns the intersection and the server learns nothing about the client set.

\Cref{fig:overview} visualizes the steps of the protocol and the interaction between the client and server.
In the base case (no variants), the Server Computation stage simply computes the intersection.
This can be extended to other functionalities such as labelled PSI, computing a function on top of the intersection, and laconic PSI.
See \Cref{sec:cpsi-variants} for such variants. \Cref{tab:notation} summarizes all the notations used in the protocol description.

\begin{table}[t]
\begin{center}
\caption{Summary of notation for algorithm description. }
\label{tab:notation}
\begin{tabular}{p{0.7 cm}l}
  % \toprule
  % \textbf{Symbol} & \textbf{Definition}   \\
  \toprule
  \multicolumn{2}{c}{\textbf{Input Parameters}}\\
  \midrule
  $\mathbb{X}$ & Client set\\
  $\mathbb{Y}$ & Server set\\
  $m$ & Client set size\\
  $n$ & Server set size\\
  $\lambda$ & Bitlength of elements \\
  $\alpha$ & Maximum Error rate \\
  \midrule
  \multicolumn{2}{c}{\textbf{Cryptographic Parameters}}\\
  \midrule
  $N$ & Polynomial modulus degree\\
  $q$ & Coefficient modulus \\
  $t$ & Plaintext modulus \\
  % $\texttt{sk}_c$ & Client Secret Key\\
  \midrule
  \multicolumn{2}{c}{\textbf{Binning Parameters}}\\
  \midrule
  $\numbins$ & Number of bins in $T_c$ and $T_s$ \\ % (index $k\in[\numbins]$)\\
  $\maxbinclient$ & Client maximum bin load \\ % (index $i\in[\maxbinclient]$)\\
  $\maxbinserver$ & Server maximum bin load \\ % (index $i'\in[\maxbinserver]$)\\
  $\effectivebitlength$ & Effective bitlength (used in PBH)\\
  \midrule
  \multicolumn{2}{c}{\textbf{Constant-weight Code Parameters}}\\
  \midrule
  $h$ & Hamming weight of constant-weight code \\
  $\ell$ & Constant-weight code length \\ % (index $j\in[\ell]$)\\
  \midrule
  \multicolumn{2}{c}{\textbf{Auxiliary Notation}}\\
  \midrule
  $T_c$ & Client table with $\numbins$ bins with $\maxbinclient$ elements\\
  $T_s$ & Server table with $\numbins$ bins with $\maxbinserver$ elements\\
  $T_c'$ & $T_c$ with \textbf{encoded} elements\\
  $T_s'$ & $T_s$ with \textbf{encoded} elements\\
  % $T'_c[k][i][j]$ & The $j$th bit of the $i$th \textbf{encoded} element in the $k^{th}$ bin of $T_c$\\
  % $T'_s[k][i][j]$ & The $j$th bit of the $i$th \textbf{encoded} element in the $k^{th}$ bin of $T_s$\\
  $pt_c$ & Plaintexts of clients elements\\
  $ct_c$ & Ciphertexts of clients elements\\
  $val_c$ & Ciphertexts of clients values\\
  $pt_s$ & Plaintexts of servers elements\\
  $val_s$ & Plaintexts of server values \\
  \bottomrule
 \end{tabular}
\end{center}
\end{table}

\subsection{Dataset Preprocessing} \label{sec:data-prep}
In the first stage, the client and server must asynchronously preprocess their datasets to prepare them for the next stage.
Preprocessing the dataset reduces the protocol latency when a client queries a server.
This preprocessing is also helpful for the server, which has a large dataset and can reuse the preprocessed dataset for many queries by various clients.
Both parties, client and server, preprocess their data in three steps: hashing optimization, encoding elements, and finally, batching (and encryption).
\Cref{fig:preprocessing} visualizes the steps of the preprocessing.
Note the actions in each step by a party are dependent on whether they are in the role of client or server.
We provide a high-level overview here with relevant pseudocode in \Cref{sec:data-prep-code}.

\begin{figure*}[t]
    \centering
    \includegraphics[width=\textwidth]{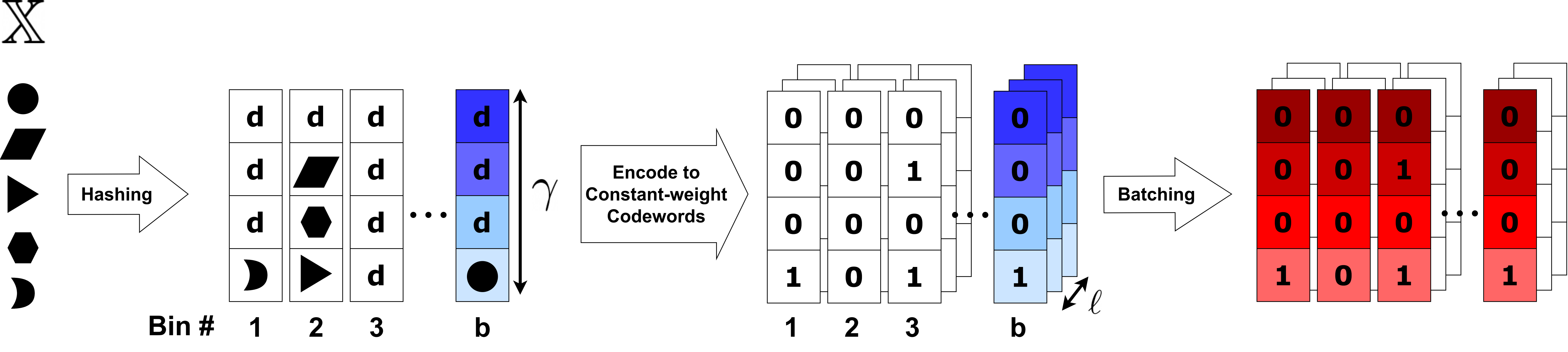}
    \caption{Client Dataset Preprocessing in \cpsi{}. $d$ indicates dummy elements, $b$ is the number of bins, and $\gamma$ is the clients maximum bin load.
    From left to right: symbols represent the real-valued payload that is hashed into bins with a maximum bin load $\gamma$. The payload is encoded into bits using our constant-weight codewords, whereby the same color indicates the same payload, and finally encrypted into a ciphertext by batching across bins. The preprocessing outputs $\ell \cdot \gamma$ ciphertexts.}
    \label{fig:preprocessing}
\end{figure*}

\emph{Hashing Optimization.}
Our hashing optimization uses a hashing-to-bins strategy in conjunction with permutation-based hashing. 
Hashing-to-bins reduces the number of elements we need to compare, and permutation-based hashing reduces the size of elements that are inserted in the table (recall $\effectivebitlength=\lambda-c$ as the effective bitlength from Section~\ref{sec:background}). 
The selected hashing-to-bins strategy in \cpsi{} must insert elements into $\numbins$ bins and does not require a stash.
Elements that fall into the stash must be compared with all other elements, which severely diminishes the performance of our protocol.
Moreover, permutation-based hashing cannot be used in conjunction with a stash.
Hence, for the binning portion of the hashing optimization, we can use known hashing strategies that do not have a stash, such as Cuckoo hashing without a stash~\cite{pinkas2019efficientcircuitbased, pinkas2018scalable} or Dual hashing~\cite{kacsmar2020differentially}.

The client and server have predetermined maximum bin loads, $\maxbinclient$ and $\maxbinserver$, respectively, which are a function of their set sizes and the number of bins.
After inserting elements into the bins, each party pad bins with dummy elements up to the maximum load of the bin.
If more elements are inserted into a bin than the maximum load, the protocol fails.
We denote the client's binning table with $\numbins$ bins and a maximum load of $\maxbinclient$ as $T_c$.
Similarly, we denote the server's binning table with $\numbins$ bins and a maximum load of $\maxbinserver$ as $T_s$.
% To maximize efficiency, we set the number of bins to be a multiple of $N$, the size of the FV plaintext. 

\emph{Encode to Constant-weight Codewords.}
The client encodes each $\effectivebitlength$-bit element in $T_c$ as a constant-weight codeword of length $\ell$ and Hamming weight $\hw$.
The Hamming weight can be chosen freely but affects the communication-computation tradeoff (see \Cref{sec:analysis}).
We set $\ell$ as a function of $\effectivebitlength$ and $\hw$ such that there is a unique representation for every $\effectivebitlength$-bit element in $T_c$.
So, 
\begin{align}\label{eq:code-length}
    \ell = \ell(\effectivebitlength, \hw) = \min \left\{ \ell\in\NN\ |\ \binom{\ell}{\hw} \geq 2^{\effectivebitlength} \right\}.
\end{align} 

We denote $T_c'$ as the new table which holds the encoded elements. 
Dummy elements are encoded to the all-zero string of length $\ell$.
The server constructs $T_s'$ from $T_s$ using the same procedure as the client uses to construct $T_c'$ from $T_c$.

\emph{Batching (and Encryption).}
In the last step, the client restructures $T_c'$ to obtain a table of plaintexts.
Recall that plaintexts are vectors of length $N$.
Each position in $pt_c[i][j]$ corresponds to elements from one of the $\numbins$ bins.
More concretely, $pt_c[i][*]$ contains the bits of the $i$-th encoded element in each of the $\numbins$ bins of $T_c$. We call this the client's $i$-th batch.
Moreover, this process, which we call \textit{batching}, allows the protocol to perform comparisons in different bins simultaneously.
Each shade of red in \Cref{fig:preprocessing} shows the contents of one plaintext.
The client additionally encrypts each plaintext with its secret key.
The server constructs plaintexts $pt_s$ in a similar fashion from the contents of $T_s'$ and the server's $i'$-th batch is defined similarly. The server does not need to encrypt its plaintexts.

We assume $\numbins=N$ in this example to simplify the explanation, but we show how to relax this assumption in relevant steps of the protocol.

\subsection{Server Computation}
The client sends the ciphertexts it has produced to the server for the next stage.
In this stage, the server receives the ciphertexts of the client's elements and performs the set intersection (see \Cref{alg:server-comp}).
At a high level, for each $i\in [\maxbinclient]$ and $i' \in [\maxbinserver]$, the server compares batch $i$ of the client's dataset with batch $i'$ of the server's dataset. The output of this step is a table of ciphertexts of size $\maxbinclient\times\maxbinserver$, which we denote by $ct_{eq}$. For $i\in [\maxbinclient]$ and $i' \in [\maxbinserver]$,  $ct_{eq}[i][i']$ denotes the result of comparing elements from batch $i$ from the client and batch $i'$ from the server.

Recall that the arithmetic constant-weight equality operator does not depend on the operands of the comparison. Hence, we can use it as a SIMD operator in line 4 of \Cref{alg:server-comp}.

\begin{algorithm}
    \caption{Compute intersection of client and server set}
    \label{alg:server-comp}
    \begin{algorithmic}[1]
        \Procedure{Intersect}{$ct_c, pt_s$}
            \ForEach {client batch $i\in[\maxbinclient]$}
                % \Comment{Intersection of the $i$th client batch and $T'_s$}	
                \ForEach{server batch $i'\in[\maxbinserver]$} \label{alg:beginloop}
                    \State $ct_{eq}[i][i'] \leftarrow \textsc{Arith-CW-Eq}(ct_c[i], pt_s[i'])$
                \EndFor \label{alg:eqloop}
                \State $ct_{ind}[i] \leftarrow \sum_{i'\in[\maxbinserver]} ct_{eq}[i][i']$
            \EndFor
            \State \textbf{return} $ct_{ind}$
        \EndProcedure
    \end{algorithmic}
\end{algorithm}

If $\numbins\neq N$, line 4 of \Cref{alg:server-comp} is repeated $\ceil{\numbins/N}$ times.

\subsection{Result Extraction}
The server sends the ciphertexts from the previous stage to the client for the final stage.
In this final stage, the client decrypts the ciphertexts received from the server to obtain the indicator vector. 
The indicator vector specifies whether there was a match in the server set, for each element in the client set.
From that, the client extracts the intersection. \Cref{alg:extract-result} shows the procedure for this stage.

\begin{algorithm}
    \caption{Extract intersection from server output}\label{alg:extract-result}
    \begin{algorithmic}[1]
        \Procedure{ExtractIntersection}{}
            \ForEach {client batch $i\in[\maxbinclient]$}
                \State $ind[i] \leftarrow \dec(ct_{ind}[i],\texttt{sk}_c)$
                \ForEach{bin $k \in [b]$}
                    \If{$ind[i][k] = 1$}
                        \State Add $T_c[k][i]$ to intersection
                    \EndIf
                \EndFor
            \EndFor
        \State \textbf{return} intersection
        \EndProcedure
    \end{algorithmic}
\end{algorithm}

% \subsection{Generalizing from Previous Assumptions}\label{sec:extensions}
% \paragraph{Variable Number of Bins.}
% For simplicity of the algorithm description, we assumed that $b=N$.
% However, this is not a constraint of the protocol.
% To utilize the entire plaintext space we choose parameters such that $b$ is always a multiple of $N$, i.e., $b\ge b_0 N$, for some $b_0$.
% If $b = b_0 N$, we use multiple plaintexts. Particularly, in line 4 of \Cref{alg:server-comp}, we simply use $\ceil{\frac{b}{N}}$ vectors of length $N$ instead of just one.

\section{\cpsi{} Variants}
\label{sec:cpsi-variants}
% \TODO{Add preanble and clarify what these things actually are and their significance}

% There are several ways that we can adapt or extend PEPSI. We descibe the variants below
% We use the base protocol from section 4 in all cases, but make adaptations based on specific scenarios

\Cref{sec:alg-desc} describes \cpsi{} for computing the intersection of two sets of $\lambda$-bit elements.
However, the algorithm can be adapted for several other scenarios as well.
Each adaptation adds extra steps to the protocol to achieve extra functionality or enhance performance.
We describe these adaptations in detail in this section.

\subsection{Optimization for Large Elements}

In some applications, the client and server elements are very large, e.g., file names or 256-bit strings.
In these cases, it is inefficient to compare large elements homomorphically.
Instead, we map each element to a smaller $\lambda$-bit element using a hash function.
However, if $\lambda$ is too small, there may be colliding elements, which results in an incorrect result.
Hence, we derive the failure rate of \cpsi{} if we map $m$ client elements and $n$ server elements to $\lambda$-bit elements.
Conversely, we choose $\lambda$ such that the failure rate of \cpsi{} is less than a given parameter $\alpha$.

Not all collisions result in failure of the protocol, i.e., an incorrect result.
Failure occurs when two unequal elements, one from the server and one from the client, have an identical mapping and are compared to each other by being placed in the same bin.
Such an event produces an incorrect result since two unequal elements match.
The upper bound on the failure rate is given in \Cref{lemma:fail-collide}.
Note that collision between two client elements (similarly, server elements) does not result in an incorrect result since client elements are not directly compared with each other.

\begin{restatable}{lemma}{failcollide}
\label{lemma:fail-collide}
When hashing $\clientsetsize$ client elements and $\serversetsize$ server elements to $\lambda$-bit strings, the probability of failure in the protocol due to collisions is upper bounded by
\begin{align}
    \frac{\numbins \maxbinclient \maxbinserver}{2^\lambda},
\end{align}

where $\numbins$, $\maxbinclient$, and $\maxbinserver$ are the number of bins, maximum client bin size, and maximum server bin size, respectively.
\end{restatable}

We prove \Cref{lemma:fail-collide} in \Cref{sec:proof-lemma}. Using \Cref{lemma:fail-collide}, we can see that, to bound the failure rate by $2^{-\alpha}$, we must choose $\lambda$, such that $\numbins \maxbinclient \maxbinserver 2^{-\lambda} \leq 2^{-\alpha}$.
Hence, we set $\lambda=\alpha+\ceil{\log_2(b \maxbinclient \maxbinserver)}$.

\subsection{Labelled PSI}
\label{sec:labelled-psi}

\cpsi{} can be adapted to labelled PSI by altering the server computation and the result extraction.
\Cref{alg:server-comp-labelled} shows the adapted algorithm.
In this algorithm, we denote the server labels as $val_{s}$, which is a vector of $\maxbinserver$ plaintexts.
Moreover, the algorithm assumes that the label fits within one plaintext slot, but this can be extended to labels of arbitrary size.
We show this extension in \Cref{sec:labelled-large} and the corresponding result extraction procedure in \Cref{alg:server-comp-labelled}.

\begin{algorithm}
    \caption{Server Computation and Result Extraction for Labelled PSI}
    \label{alg:server-comp-labelled}
    \begin{algorithmic}[1]
        \Procedure{LabelledPSI}{$ct_c[i]$,$pt_s[i']$}
            \ForEach {client batch $i\in[\maxbinclient]$}
                \ForEach{server batch $i'\in[\maxbinserver]$}
                    \State $ct_{eq}[i][i'] \leftarrow \textsc{Arith-CW-Eq}(ct_c[i], pt_s[i'])$
                \EndFor 
                \State $ct_{res}[i] \leftarrow \sum_{i'} val_s[i'] \cdot ct_{eq}[i][i']$
            \EndFor
            \State Send $ct_{res}$ to the client
        \EndProcedure
        \vspace{2mm}
        \Procedure{ExtractLabels}{$ct_{res}$}
            \ForEach {client batch $i\in[\maxbinclient]$}
                \State $val[i] \leftarrow \dec(ct_{res}[i],\texttt{sk}_c)$
                \ForEach{bin $k \in [b]$}
                    \If{$val[i][k] \neq 0$}
                        \State Add $val[i][k]$ to intersection labels
                    \EndIf
                \EndFor
            \EndFor
            \State \textbf{return} intersection labels
        \EndProcedure
    \end{algorithmic}
\end{algorithm}

\subsection{Circuit PSI}
\label{sec:circuit-psi}

\cpsi{} can additionally be extended to compute functions over the intersection.
Such functions include PSI-Cardinality, PSI-Sum~\cite{kacsmar2020differentially}, PSI-Inner-Product~\cite{ion2020deploying,pjac}, and beyond. We introduce additional functionalities in the appendix.

\emph{PSI-Sum \& PSI-Cardinality.}
In PSI-sum, the sum of server values associated with the elements in the intersection are output to the client. \Cref{alg:server-comp-sum} show the modified server computation and result extraction for PSI-sum, respectively.
Kacsmar et al. use the same algorithm for securely computing PSI-sum and give proof of correctness and security~\cite{kacsmar2020differentially}.
PSI-Cardinality, which computes the size of the intersection, is a special case where the server values are equal to one.

\begin{algorithm}
    \caption{Server computation and result extraction for PSI-Sum}\label{alg:server-comp-sum}
    \begin{algorithmic}[1]
        \Procedure{ComputePSISum}{$ct_c, pt_s$}
            \ForEach {client batch $i\in[\maxbinclient]$}
                % \InlineComment{Intersection of the $i$th client batch and $T'_s$}	
                \ForEach{server batch $i'\in[\maxbinserver]$} \label{alg:sumbeginloop}
                    \State $ct_{eq}[i][i'] \leftarrow \textsc{Arith-CW-Eq}(ct_c[i], pt_s[i'])$
                \EndFor \label{alg:sumeqloop}
            \EndFor
            \State $ct_{sum} \leftarrow \sum_{i'\in[\maxbinserver]} val_s[i'] \cdot \sum_{i} ct_{eq}[i][i']$
            \State Sample vector $r \leftarrow [r_0, r_1, \cdots, r_{N-1}]$ s.t. $\sum r_i = 0$
            \State $ct_{res} \leftarrow ct_{sum} + r$
            \State \textbf{return} $ct_{res}$
        \EndProcedure
    \vspace{3mm}
        \Procedure{ExtractPSISum}{$ct_{res}$}
            \State $res \leftarrow \dec(ct_{res},\texttt{sk}_c)$
            \State $ct_{sum} \leftarrow \sum_{k\in[b]} res[i]$
            \State \textbf{return} $ct_{sum}$
        \EndProcedure
    \end{algorithmic}
\end{algorithm}

\emph{PSI-Inner-Product.}
PSI-sum can be generalized to PSI-Inner-Product~\cite{pjac} as well. 
Assuming that the encrypted client values are denoted as $val_c$, we simply replace line 4 of \Cref{alg:server-comp-sum} with 
$$
    sum \leftarrow \sum_{i'\in[\maxbinserver]} val_s[i'] \cdot \sum_{i} val_{c}[i] \cdot ct_{eq}[i][i'].
$$

\section{Analysis of \cpsi{}}
\label{sec:analysis}

In this section, we provide an overview of the security analysis and details on communication and computation complexities.
% While the security analysis is rather straightforward, the communication and computation costs are more intricate.
The communication and computation analysis is necessary to choose optimal parameters and enables \cpsi{} to achieve competitive performance compared to related work.

\paragraph{Security Analysis.}
The \cpsi{} protocol operates in the semi-honest model, i.e., the client and server follow the protocol but may try to infer extra information.
The client input privacy is guaranteed due to the use of homomorphic encryption.
The noise level in the homomorphic ciphertexts can reveal extra information about the server's dataset.
Hence, we need a technique to make the noise level of the output ciphertexts indistinguishable from fresh ciphertexts.
This is referred to as \textit{circuit privacy} in the literature~\cite{bourse2016fhe}.
Techniques such as noise flooding~\cite{chen2019multi, asharov2012multiparty} achieve this property and can also be used in this work.
Noise flooding is not implemented in the available homomorphic encryption libraries, so we do not perform that step in our implementation.
However, it has a negligible effect on runtime, i.e., at the cost of one extra homomorphic addition, so our experimental results are unaffected. 
We choose all parameters such that we have 128-bit security.

\emph{Extending to the Malicious Model.}
The malicious security model removes the assumption that adversaries follow the protocol and considers active adversaries behaving arbitrarily.
The malicious model for PSI has two limitations:
1) A malicious party may simply substitute its input and learn (parts of) the other party's set once the intersection is revealed, which cannot be prevented in the malicious model.
2) In a one-round protocol, arbitrary behaviour is only possible in the input-carrying first message.
Hence, additional information can only be computed from the message (output) received in response to this first message.
Using techniques such as input verification~\cite{certifiedsets,authpsi} and verifiable homomorphic encryption~\cite{viand2023verifiable} we can augment our protocol to account for malicious clients and servers. However, we leave a details description of such a protocol for future work.

\paragraph{Notes on Performance Analysis.}
The communication and computation costs of \cpsi{} depend on many parameters, the most apparent of which are the set sizes and the bitlength of elements.
Two parameters that can be tuned to optimize performance are the binning strategy and the Hamming weight.

The choice of binning strategy and associated parameters has been extensively discussed in previous works~\cite{cuckoohashing2004, pinkas2018scalable}, and using existing analysis, we can derive the relevant binning parameters $\numbins$, $\maxbinclient$, and $\maxbinserver$, given the set sizes.
Hence, in this work, we compute the communication and computation cost as a function of $\maxbinclient$, $\maxbinserver$, and $\numbins$ instead of the set sizes.
As a result, our analysis is agnostic to the binning strategy.

The other parameter we can optimize over is the Hamming weight.
The Hamming weight $\hw$ is a parameter of the constant-weight equality operator.
Cryptographic parameters depend solely on $\hw$ as it determines the multiplicative depth.
The code length is a function of the bitlength (and effective bitlength) and the Hamming weight, so it does not need to be optimized.
In this section, we discuss how the communication and computation of \cpsi{} are affected by the Hamming weight and the tradeoff between these two metrics.

\paragraph{Analysis for Parameter Selection.}
Using the analysis from this section, we define strategies for optimally selecting parameters for \cpsi{}.
There are two obstacles to optimizing the parameters in \cpsi{}.
Firstly, there are circular dependencies between the parameters. For example, the cryptographic parameters, $N$ and $q$ dictate the number of bins, which influences the error rate which in turn influences the parameters of the constant-weight code.
However, the Hamming weight of the constant-weight code determines the multiplicative depth, which puts a lower bound on the cryptographic parameters.
The second obstacle is that there is no precise closed-form formula for the runtime so it requires experimental evaluation.
Moreover, in many cases, there is no definitive choice for the parameters that optimize both communication and computation.
In such cases, we visualize the communication computation tradeoff to assist in parameter selection.
Hence, optimizing the parameters amounts to a non-trivial task that requires experimental evaluation and careful and systemic selection of the parameters. 

\subsection{Communication Complexity}
The total communication complexity of \cpsi{} consists of the request and response messages.
The response message depends on the function that is being evaluated. 
For example, if the set intersection is returned, the response is as big as the request message.
In contrast, if we want the intersection cardinality, the response is only one ciphertext.
In all the examples examined in this work, the response is smaller than the request size. 
Hence, to encompass all of these applications, we analyze and optimize the request size.

\paragraph{Asymptotic Communication Complexity.}
\Cref{thm:asymp-comm} derives the asymptotic complexity of the request in \cpsi{} as a function of the code length, $\ell$.

\begin{theorem} \label{thm:asymp-comm}
    The asymptotic complexity of the request message in \cpsi{} is $O(\numbins\cdot\maxbinclient\cdot\ell)$, where $\ell=\ell(\effectivebitlength, \hw)$ as defined in \Cref{eq:code-length}.
\end{theorem}

\begin{proof}
As shown in \Cref{sec:alg-desc}, in the request, the encrypted plaintexts containing $T'_c$ are sent to the server. $T'_c$ is a table with dimension $\numbins \times \maxbinclient \times \ell$, with one bit in each cell of this table. Hence, the total communication is $O(\numbins\cdot\maxbinclient\cdot\ell)$.
\end{proof}

\begin{corollary}
    If we use Cuckoo hashing in \cpsi{}, then the communication complexity can be simplified to $O(\clientsetsize \cdot \ell)$ where $m$ is the client set size and $\ell=\ell(\effectivebitlength, h)$ as defined in \Cref{eq:code-length}.
\end{corollary}

\paragraph{Concrete Communication Cost.} 
\Cref{thm:asymp-comm} only states the asymptotic complexity of the request size, but in practice, the size of the ciphertext must also be considered.
Encoded elements in the client's set are put into $\ceil{b/N}\cdot\maxbinclient\cdot\ell$ ciphertexts and each ciphertext in the request is approximately $N\cdot q$ bits, where $q$ is the FV ciphertext modulus which is used.
So overall, the size of the request in \cpsi{} is $\ceil{\numbins/N}\cdot\maxbinclient\cdot\ell\cdot N\cdot q$ bits (or simply $\numbins\cdot\maxbinclient\cdot\ell\cdot q$ when $\numbins$ is a multiple of $N$).

Note that the ciphertext modulus depends on the multiplicative depth, which depends only on the Hamming weight.
$N$ and $q$ must also satisfy requirements for security which have been outlined in the literature~\cite{HomomorphicEncryptionSecurityStandard}.
For each Hamming weight, we select the smallest parameter set, which ensures the correctness of the result while also providing at least 128-bit security.
\Cref{tab:mod-per-hamming-weight} shows the polynomial modulus degree and bitlength of the coefficient modulus we use for each $\hw$.

\begin{table}[H]
    \centering
    \caption{Polynomial modulus degree and coefficient modulus as a function of the Hamming weight.}
    \label{tab:mod-per-hamming-weight}    
    \begin{tabular}{c|c|c|c|c|c|c|c}
        \toprule
        $\hw$ & 1 & 2 & 3-4 & 5-8 & 9-16 & 17-32 & 33-64 \\
        \midrule
        $\log_2 N$ & 12 & 13 & 13 & 13 & 14 & 14 & 14 \\
        $\log_2 q$ & 72 & 144 & 168 & 204 & 240 & 276 & 312 \\
        \bottomrule
    \end{tabular}
\end{table}

\paragraph{Optimizing Communication Cost.}
The communication cost is derived as a function of the code length.
However, the code length itself depends on the Hamming weight and effective bitlength, as shown in \Cref{eq:code-length}.
Hence, to minimize the communication complexity, we look for the Hamming weight which minimizes $\ceil{\numbins/N}\cdot\maxbinclient\cdot\ell\cdot N\cdot q$.

We know that $\maxbinclient$ does not depend on the Hamming weight so it acts as a constant in our optimization and can be removed.
So we can minimize $\ceil{\numbins/N} \cdot \ell \cdot N\cdot q$, where $\ell=\ell(\effectivebitlength, \hw)$.
The effect of $\numbins$ is more convoluted and must be taken into account in the optimization.
For simplicity, we assume $\numbins\leq4096$, which is less than the smallest possible value for $N$, so $\ceil{\numbins/N}=1$. We show results for other values of $\numbins$ in the appendix.

\Cref{fig:ell-per-ex-4096} plots $\ceil{\numbins/N} \cdot\ell\cdot N \cdot q$ with $\numbins\leq 4096$ as a function of the Hamming weight for $\effectivebitlength=16,32,48$ to demonstrate that a minimum exists.
The minimum communication occurs for $\hw=8, 8, 23$, respectively.

We plot the optimal Hamming weight for communication for each effective bitlength in \Cref{fig:best-hamming-weight} assuming $\numbins\leq 4096$.
In the case of a tie between two Hamming weights, we choose the smaller Hamming weight since it requires less computation.
Generally, as the effective bitlength grows, a larger Hamming weight is required to achieve the smallest communication cost.

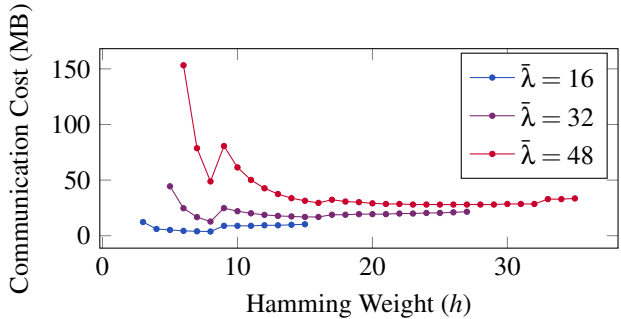
\begin{figure}[H]
    \centering
    % \subfloat{    
    \begin{tikzpicture}
        \begin{axis}[
            table/col sep=comma,
            xlabel={Hamming Weight ($\hw$)},
            ylabel={Communication Cost (MB)},
            width=\columnwidth,
            height=0.5\columnwidth
        ]
            \addplot [mark=*, mark size=1pt, mark options={color=pepsiblue}, color=pepsiblue] table [y expr=\thisrow{CodeLengthWeight}/\BMB, x=HammingWeight]{data/output-bt=16-b=4096.csv};
            \addlegendentry{$\effectivebitlength=16$}
            \addplot [mark=*, mark size=1pt, mark options={color=pepsipurple}, color=pepsipurple] table [y expr=\thisrow{CodeLengthWeight}/\BMB, x=HammingWeight]{data/output-bt=32-b=4096.csv};
            \addlegendentry{$\effectivebitlength=32$}
            \addplot [mark=*, mark size=1pt, mark options={color=pepsired}, color=pepsired] table [y expr=\thisrow{CodeLengthWeight}/\BMB, x=HammingWeight]{data/output-bt=48-b=4096.csv};
            \addlegendentry{$\effectivebitlength=48$}
        \end{axis}
    \end{tikzpicture}
    \caption{Code length as a function of the Hamming weight for $\effectivebitlength\in\{16, 32, 48\}$ for $\numbins\leq 4096$. The minimum occurs for a Hamming weight of 8, 8, and 23, respectively.}
    \label{fig:ell-per-ex-4096}
\end{figure}

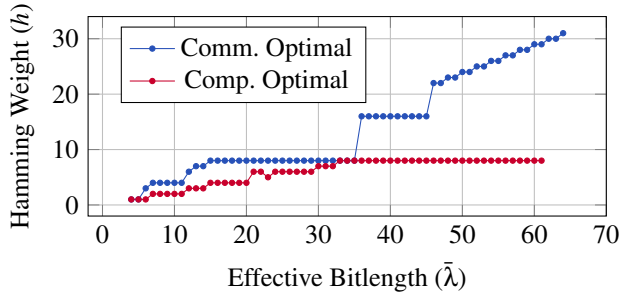
\begin{figure}[H]
    \centering
    \begin{tikzpicture}
        \begin{axis}[
            table/col sep=comma,
            xlabel=Effective Bitlength ($\effectivebitlength$),
            ylabel={Hamming Weight ($\hw$)},
            xmajorgrids,
            ymajorgrids,
            legend style={at={(0.3,0.95)}, anchor=north,legend columns=1},
            width=\columnwidth,
            height=0.5\columnwidth
        ]
            \addplot[mark=*, mark size=1pt, mark options={color=pepsiblue}, color=pepsiblue] table [y=BestHammingWeight, x=EffectiveBitlength]{data/best-hamming-b=4096.csv};
            \addlegendentry{Comm. Optimal}
            \addplot[mark=*, mark size=1pt, mark options={color=pepsired}, color=pepsired] table [y=BestHammingWeight, x=EffectiveBitlength]{data/best-hamming-comp-b=4096.csv};
            \addlegendentry{Comp. Optimal}
        \end{axis}
    \end{tikzpicture}
    \caption{
        Hamming weight which optimizes communication and computation as a function of the effective bitlength ($\effectivebitlength$) in blue and red, respectively. We assume $\numbins\leq 4096$ in these graphs.
    }
    \label{fig:best-hamming-weight}
\end{figure}

\subsection{Computation Complexity}

Amongst the stages of the protocol, the Result Extraction has a small computation overhead and also depends on the specific function that we derive. The data preparation can be done offline. 
Most of the online computation time is dedicated to the Server Computation stage. Within that stage, the majority of the runtime is dedicated to computing the encrypted indicator vector (\Cref{alg:server-comp}) and is required for any subsequent computation.
The server computation algorithms described in \Cref{sec:circuit-psi} may perform additional operations after the indicator vector is derived, but in all cases, these operations are insignificant compared to calculating the indicator vector.
Hence, we base our analysis and parameter selection on \Cref{alg:server-comp} and calculate and optimize the runtime of this algorithm. 

\paragraph{Approximate Server Computation Complexity.}
We approximate the computational complexity of the Server Computation stage by counting the number of expensive homomorphic operations, i.e., plaintext and homomorphic multiplications.

\begin{theorem}\label{eq:server-comp}
    The number of homomorphic operations in \Cref{alg:server-comp} is 
    % \begin{align}
        $( \ell \cdot \texttt{PM} + \hw \cdot \texttt{M} )\cdot \ceil{b/N} \cdot \maxbinclient \cdot \maxbinserver $
    % \end{align}
    where \texttt{PM} and \texttt{M} denote plaintext and homomorphic multiplication, respectively, and $\ell=\ell(\effectivebitlength, \hw)$.
\end{theorem}

\begin{proof}
The constant-weight equality operator in line 3 of \Cref{alg:server-comp} consists of $\ell$ homomorphic multiplications and $\hw$ plaintext multiplications~\cite{mahdavi2022constant}.
In the general case when the number of bins is larger than $N$, line 4 of \Cref{alg:server-comp} is performed $\ceil{\numbins/N}$ times.
Moreover, this line is repeated in two loops of length $\maxbinclient$ and $\maxbinserver$, hence the total number of operations is $( \ell \cdot \texttt{PM} + \hw \cdot \texttt{M} ) \cdot\ceil{\numbins/N} \cdot \maxbinclient \cdot \maxbinserver$.
\end{proof}

\paragraph{Optimizing Computation Cost.}
Given that we do not have a closed-form formula for the computation complexity of the homomorphic operators, optimizing the computation complexity is not possible.
However, the concrete computation cost of \cpsi{} can be optimized empirically.
Optimization is performed over the Hamming weight since we assumed that the bitlength, $\numbins$, $\maxbinclient$, and $\maxbinserver$ are given.

The client and server set sizes determine $\maxbinclient$ and $\maxbinserver$, based on the binning strategy.
The choice of $\hw$ does not affect any of these parameters and only influences $\ell$.
Hence, $\maxbinclient$ and $\maxbinserver$ can be treated as constants in the optimization and it suffices to minimize $\ceil{\frac{\numbins}{N}} \cdot (\ell \cdot \texttt{PM} + \hw \cdot \texttt{M})$.
Theoretically, optimizing this term is difficult since there are no closed-form formulas for the runtime of homomorphic operations.
Instead, we empirically optimize it by trying different Hamming weights.
We set $\maxbinclient=\maxbinserver=1$ to eliminate the variance introduced by those parameters and run \cpsi{} using only one thread.

\begin{figure}[H]
    \centering
    % \subfloat{    
    \begin{tikzpicture}
        \begin{axis}[
            table/col sep=comma,
            xlabel={Hamming Weight ($\hw$)},
            ylabel={Runtime (ms)},
            width=\columnwidth,
            height=0.63\columnwidth
        ]
            \addplot [mark=*, mark size=1pt, mark options={color=pepsiblue}, color=pepsiblue] table [y=runtime, x=hw]{data/output-comp-bt=16-b=4096.csv};
            \addlegendentry{$\effectivebitlength=16$}
            \addplot [mark=*, mark size=1pt, mark options={color=pepsipurple}, color=pepsipurple] table [y=runtime, x=hw]{data/output-comp-bt=32-b=4096.csv};
            \addlegendentry{$\effectivebitlength=32$}
            \addplot [mark=*, mark size=1pt, mark options={color=pepsired}, color=pepsired] table [y=runtime, x=hw]{data/output-comp-bt=48-b=4096.csv};
            \addlegendentry{$\effectivebitlength=48$}
        \end{axis}
    \end{tikzpicture}
    % }
    \caption{Runtime as a function of the Hamming weight for $\effectivebitlength\in\{16, 32, 48\}$ for $\numbins\leq 4096$. The minimum occurs for a Hamming weight of 4, 8, and 8, respectively.}
    \label{fig:ell-per-ex-comp-4096}
\end{figure}
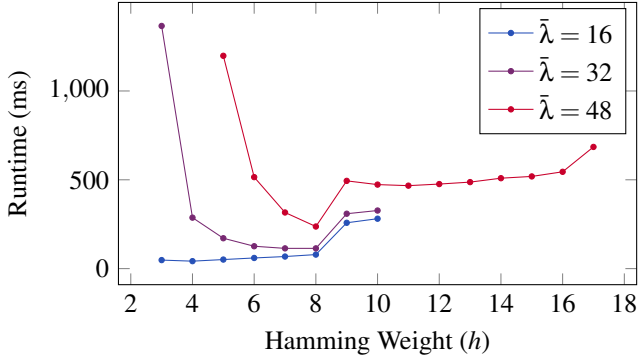

\Cref{fig:best-hamming-weight} shows, in red, the Hamming weights which optimize computation time for each effective bitlength assuming that $\numbins\leq 4096$.
Similar to the communication, the optimal Hamming weight for runtime increases as $\effectivebitlength$ increases.

\subsection{Communication/Computation Tradeoff}
Note that choosing parameters that optimize the computation would increase communication costs and vice versa.
In practice, it may be reasonable to choose parameters that fall somewhere between the two extremes. \Cref{fig:tradeoff} visualizes the communication/computation tradeoff for the intersection as the Hamming weight varies.
We plot the communication/computation tradeoff for $\effectivebitlength=16,24,32$ assuming that $\numbins \leq 4096$.
% We can see that the communication and computation optima occur at $h=8$ and $h=4$, respectively, for $\effectivebitlength=16, 20$ which is aligned with \Cref{fig:best-hamming-weight}.
% For $\effectivebitlength=24$, the communication and computation optima occur at $h=8$ and $h=4$, respectively.

\begin{figure}[H]
    \centering
    \begin{tikzpicture}[]
        \begin{axis}[
            table/col sep=comma,
            xlabel={Communication (MB)},
            ylabel={Runtime (ms)},
            xmajorgrids,
            ymajorgrids,
            colormap={redblue}{color=(pepsired) color=(pepsiblue)},
            colorbar,
            colorbar style={
                title={Hamming weight},
                ytick={3,4,...,9,10},
                title style={at={(2.2,0.7)}, anchor=west, rotate=-90} % <-- position and rotation of the colorbar title
            },
            % legend pos=inner north east,
            width=0.8\columnwidth,
            height=0.7\columnwidth
        ]
            \addplot[mesh, line width=1pt, scatter, scatter src=explicit, point meta=explicit, scatter/use mapped color={draw=mapped color, fill=mapped color, mark=triangle*}] table [y=runtime, x=comm, meta=hw]{data/tradeoff-16-filtered.csv};
            % \addlegendimage{mark=triangle*}
            % \addlegendentry{$\effectivebitlength=16$}

            % \addplot[mesh, line width=1pt, scatter, scatter src=explicit, point meta=explicit, scatter/use mapped color={draw=mapped color, fill=mapped color, mark=square*}] table [y=runtime, x=comm, meta=hw]{data/tradeoff-20-filtered.csv};
            % \addlegendimage{mark=square*}
            % \addlegendentry{$\effectivebitlength=20$}

            \addplot[mesh, line width=1pt, scatter, scatter src=explicit, point meta=explicit, scatter/use mapped color={draw=mapped color, fill=mapped color, mark=square*}] table [y=runtime, x=comm, meta=hw]{data/tradeoff-24-filtered.csv};
            % \addlegendimage{mark=*}
            % \addlegendentry{$\effectivebitlength=24$}

            \addplot[mesh, line width=1pt, scatter, scatter src=explicit, point meta=explicit, scatter/use mapped color={draw=mapped color, fill=mapped color, mark=*}] table [y=runtime, x=comm, meta=hw]{data/tradeoff-32.csv};

            % \addplot[mesh, line width=1pt, scatter, scatter src=explicit, point meta=explicit, scatter/use mapped color={draw=mapped color, fill=mapped color, mark=*}] table [y=runtime, x=comm, meta=hw]{data/tradeoff-48.csv};
            
        \end{axis}
    \end{tikzpicture}
    \caption{Visualizing the tradeoff between communication cost and runtime in \cpsi{} for $\effectivebitlength=16$ (triangles), $\effectivebitlength=24$ (squares), and $\effectivebitlength=32$ (circles) with $\numbins \leq 4096$. Each point indicated running \cpsi{} with a specific Hamming weight which varies from 3 (on the lower right side) to 9 (on the upper left).
    We set $\gamma=\mu=1$ in this example to isolate the effect of other parameters.
    }
    \label{fig:tradeoff}
\end{figure}
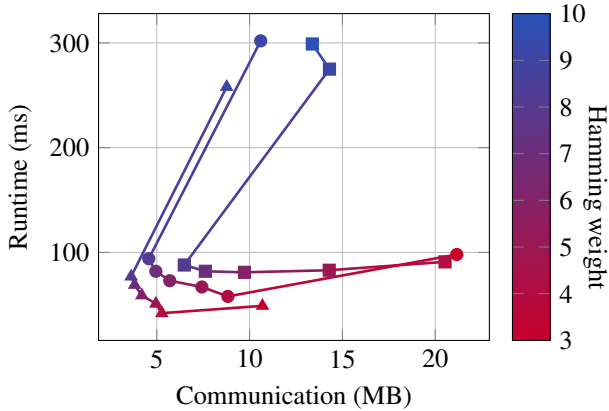

Notice that in some cases, the communication optimal may incur a high burden on the computation for very little gain, and vice versa. For example, in \Cref{fig:ell-per-ex-4096} for $\effectivebitlength=48$, the communication cost decreases very slightly (less than 6\%) as we increase $\hw$ from $16$ to $23$. However, in \Cref{fig:ell-per-ex-comp-4096}, the computation cost is much lower for $\hw=16$ compared to $\hw=23$ (over 40\%).

\section{Evaluation}
% In this section, we empirically evaluate \cpsi{} and compare it with state-of-the-art protocols in PSI, labelled PSI, and circuit PSI, with a focus on the unbalanced setting.

\paragraph{Experimentation Details.}

We implement \cpsi{} in C++ using the Microsoft SEAL\footnote{\url{https://github.com/microsoft/SEAL}} library and measure runtime and communication.
The SEAL library implements a variant of the FV cryptosystem and supports the operations we use (recall Section~\ref{sec:background}). 
Our implementation of PEPSI and the code used in our evaluation is available on Github~\footnote{\url{https://github.com/RasoulAM/pepsi}}.
We run experiments using two scenarios: 1) 32-bit elements and 2) large elements.
32-bit elements can be used to represent phone numbers in private contact discovery~\cite{kales2019mobile, kacsmar2020differentially}.
Large elements are useful for applications such as checking for compromised credentials, where credentials could be strings of arbitrary length~\cite{li2019protocols}.
Parameters for these two scenarios are selected as follows:

\emph{Parameters for 32-bit Elements.}
Based on \Cref{fig:best-hamming-weight}, the best Hamming weight we can choose for the communication cost is $\hw=8$, which results in $N=2^{13}$ and $\log_2 q = 192$.
As explained in \Cref{sec:alg-desc}, we round the number of bins up to a multiple of $N$ for efficiency.
Consequently, in our experiments, we have $\numbins=N=2^{13}$, $\effectivebitlength=19$, and $\ell=24$.
For our binning strategy, we use Cuckoo hashing with three hash functions without a stash, so $\maxbinclient=1$.
We choose the parameters of Cuckoo hashing such that the failure rate is less than $2^{-40}$.
Based on the analysis in the literature~\cite{pinkas2019efficientcircuitbased, pinkas2018scalable}, the number of bins must be at least $1.27 \clientsetsize$.
The server bin size depends on the server set size and the number of bins.
We use upper bounds for the server bin load found in the literature~\cite{cuckoohashing2004} and confirm the numbers experimentally using a simulation of the binning strategy.

\emph{Parameters for Large Elements.}
We choose the parameters such that the failure rate due to colliding elements is less than $\alpha=2^{-40}$, so $\lambda$ changes as the server set size increases.
The binning strategy is Cuckoo hashing with three hash functions and no stash, so $\maxbinclient=1$.
As a result,
$\numbins=N=16384,
\lambda=40+\ceil{\log_2(\numbins \cdot\maxbinclient\cdot \maxbinserver)} = 40+\ceil{\log_2(\numbins \cdot \maxbinserver)}, \effectivebitlength=40+\ceil{\log_2(\maxbinserver)}.$

\subsection{PSI Evaluation and Comparison}

% We focus our comparison on two non-interactive PSI protocols in the literature:
% 1) DiPSI by Kacsmar et al.~\cite{kacsmar2020differentially} and 2) the work of Cong et al.~\cite{cong2021LabelledPSI}.

\paragraph{Comparison with DiPSI~\cite{kacsmar2020differentially}.}

DiPSI~\cite{kacsmar2020differentially} is the most similar work in the literature and operates in one round, similar to \cpsi{}.
While \cpsi{} uses a homomorphic equality operator, similar to DiPSI, it improves on DiPSI in several aspects.
The two main differences are the use of a more efficient equality operator and permutation-based hashing.
These differences result in significant improvements in concrete runtime and communication costs.
Moreover, the advantage is increased by careful optimization of the parameters of \cpsi{}, as was mentioned in \Cref{sec:analysis}.
Such optimizations are not required in DiPSI given that there is no variable parameter in the equality operator.

\Cref{tab:compare-protocols} shows that the asymptotic complexity of DiPSI and \cpsi{} is identical, but if we include parameters regarding the bitlength of elements, there is a difference.
\Cref{tab:compare-circuit-psi-with-he-theory} displays the asymptotic complexity of the two protocols.
The table also shows that the multiplicative depth of \cpsi{} does not depend on the bitlength, which results from our choice of equality operator.

\begin{table}[H]
    \centering
    \resizebox{\columnwidth}{!}{
    \begin{tabular}{c|c|c|c}
    \toprule
        Protocol & Communication & Computation & Mult. Depth\\
    \midrule
        DiPSI & $O\left(m \lambda \right)$ & $O\left(n \cdot \lambda \right)$ & $\log_2 \lambda$\\
       \cpsi{} & $O\left(m \cdot \ell( \effectivebitlength, h) \right)$ & $O\left(n \cdot \ell (\effectivebitlength, h)\right)$ & 1 + $\log_2 h$ \\
    \bottomrule
    \end{tabular} 
    }
    \caption{Asymptotic complexity of DiPSI vs \cpsi{} in the unbalanced case, i.e., $m\ll n$.}
    \label{tab:compare-circuit-psi-with-he-theory}
\end{table}

All our experiments follow the unbalanced setting where the server has more elements than the client.
DiPSI performs the intersection over 32-bit elements, so we do the same for a fair comparison.
Cong et al.~\cite{cong2021LabelledPSI} allows to use elements of arbitrary length so we compare with the variant of \cpsi{} which can do the same.
We measure all messages exchanged as the total communication and for computation, we measure the server's total runtime.
We repeat experiments three times and report the average.
However, the runtime of DiPSI is prohibitively high, so we take numbers from their paper~\cite{kacsmar2020differentially}.

\Cref{tab:dipsi-cpsi} shows runtime of DiPSI and \cpsi{} for intersection of sets with 32-bit elements.
This table shows that DiPSI is much slower than our work.
\cpsi{} is at least four orders of magnitude faster than DiPSI.
Most of this speedup can be attributed to the faster comparison operator, which we use in \cpsi{}.
The communication of \cpsi{} is much smaller as well, which stems from 1) better parameters for the cryptosystem, 2) permutation-based hashing in conjunction with Cuckoo hashing, and 3) optimal Hamming weight, which minimizes communication.
\cpsi{} requires 90\% less communication compared to DiPSI.
DiPSI can extend to labelled PSI and circuit PSI as well, which increases the total runtime and communication. We omit DiPSI from the rest of the evaluation due to its low performance.

\begin{table}[t]
    \centering
    \caption{Private set intersection over 32-bit elements using DiPSI and \cpsi{}. The bin size in \cpsi{} is $b=8192$ in all cases. DNF indicates experiments which did not finish in less than an hour. * We copied the runtimes and communicating cost for DiPSI from their paper~\cite{kacsmar2020differentially}. Note that the runtime of DiPSI is in minutes.}
    \label{tab:dipsi-cpsi}    
\resizebox{\columnwidth}{!}{
    \begin{tabular}{c|c|c|c|c|c|c|c|c}
    \toprule
    \multirow{3}{*}{$n$} & \multirow{3}{*}{$m$} & \multicolumn{2}{c|}{DiPSI~\cite{kacsmar2020differentially}} & \multicolumn{5}{c}{\cpsi{}} \\
    & & Time & Comm. & $\maxbinserver$ & \multicolumn{2}{c|}{Time (s)} & \multicolumn{2}{c}{Comm. (MB)} \\
    & & \textbf{(mins)} & (MB) & & Offline & Online & Req. & Resp. \\
    \midrule
    \multirow{3}{*}{$2^{20}$} 
     & $1024$ & 600$^*$ & 200$^*$ & \multirow{3}{*}{526} & \multirow{3}{*}{0.11} & 0.83 & 4.1 & 0.11 \\
     & $2048$ & 350$^*$ & 200$^*$ & & & 0.81 & 4.1 & 0.11 \\
     & $4096$ & 190$^*$ & 200$^*$ & & & 1.5 & 8.1 & 0.22 \\
    \midrule
    \multirow{3}{*}{$2^{24}$} 
     & $1024$ & DNF & DNF & \multirow{3}{*}{6710} & \multirow{3}{*}{1.5} & 8.9 & 4.1 & 0.11 \\
     & $2048$ & DNF & DNF & & & 9.0 & 4.1 & 0.11 \\
     & $4096$ & DNF & DNF & & & 17.5 & 8.1 & 0.22 \\
    \midrule
    \multirow{3}{*}{$2^{28}$} 
     & $1024$ & DNF & DNF & \multirow{3}{*}{100565} & \multirow{3}{*}{19.9} & 133 & 4.1 & 0.11 \\
     & $2048$ & DNF & DNF & & & 133 & 4.1 & 0.22 \\
     & $4096$ & DNF & DNF & & & 260 & 8.1 & 0.22 \\
    \bottomrule
    \end{tabular}
}
\end{table}

To summarize, \cpsi{} is a strict improvement over DiPSI. It is non-interactive, similar to DiPSI, requires strictly less communication, is over three orders of magnitude faster, and is capable of all functionalities that DiPSI offers.

\paragraph{Comparison with Cong et al.~\cite{cong2021LabelledPSI}}

Cong et al.~\cite{cong2021LabelledPSI} is the state-of-the-art amongst non-interactive PSI protocols with only two rounds of interaction. 
The work of Cong et al. has strictly improved over the work of Chen et al.~\cite{chen2018LabelledPSI, chen2017fast}. 
Cong et al. denote the time required to compute the OPRF of server elements as \textit{offline} time.
We report the offline and online time separately.
We use the publicly available implementation of the work of Cong et al.~\cite{cong2021LabelledPSI}, which is published on Github\footnote{\url{https://github.com/microsoft/APSI}} and is called APSI.
We use parameters provided with their code for each client and server set size pair.
\Cref{tab:cong-cpsi} compares \cpsi{} and the work of Cong et al.~\cite{cong2021LabelledPSI} for intersection of sets with large elements.

% \begin{table*}[t]
%     \centering
%     \begin{tabular}{c|c|c|c|c|c|c|c|c|c|c}
%     \toprule
%     $n$ & $m$ & \multicolumn{4}{c|}{Cong et al.~\cite{cong2021LabelledPSI}} & \multicolumn{5}{c}{\cpsi{}} \\
%     & & \multicolumn{2}{c|}{Time (s)} & \multicolumn{2}{c|}{Comm. (MB)} & $\maxbinserver$ & \multicolumn{2}{c|}{Time (s)} & \multicolumn{2}{c}{Comm. (MB)}  \\
%     & & \small{Offline} & \small{Online} & Request & Response & & Offline & Online & Request & Response \\
%     \midrule
%     \multirow{3}{*}{$2^{20}$} 
%      & $1024$ & 5.6 & 0.44 & 1.3 & 1.2 & \multirow{3}{*}{296} & \multirow{3}{*}{0.11} & 7.8 & 30.0 & 0.22 \\
%      & $2048$ & 5.7 & 0.82 & 2.7 & 1.5 & &  & 7.1 & 30.0 & 0.22 \\
%      & $4096$ & 6.1 & 0.94 & 3.5 & 2.2 & &  & 7.0 & 30.0 & 0.22 \\
%     \midrule
%     \multirow{3}{*}{$2^{24}$} 
%      & $1024$ & 99 & 1.2 & 1.8 & 2.1 & \multirow{3}{*}{3487} & \multirow{3}{*}{1.5} & 75.4 & 33.7 & 0.22 \\
%      & $2048$ & 97 & 1.5 & 3.0 & 2.3 & &  & 75.6 & 33.7 & 0.22 \\
%      & $4096$ & 97 & 1.8 & 5.2 & 2.6 & &  & 75.9 & 33.7 & 0.22 \\
%     \midrule
%     \multirow{3}{*}{$2^{28}$} 
%      & $1024$ & 1770 & 7.3 & 3.5 & 9.5 & \multirow{3}{*}{50812} & \multirow{3}{*}{19.9} & 1144 & 39.2 & 0.22 \\
%      & $2048$ & 1720 & 7.4 & 6.0 & 9.5 & &  & 1154 & 39.2 & 0.22 \\
%      & $4096$ & 1790 & 7.7 & 8.5 & 9.8 & &  & 1141 & 39.2 & 0.22 \\
%     \bottomrule
%     \end{tabular}
%     \caption{Private set intersection over large elements using the work of Cong et al. and \cpsi{}. The number of bins is set to $b=16384$ in \cpsi{} for all cases.}
%     \label{tab:cong-cpsi}
% \end{table*}

\begin{table*}[t]
    \centering
    \begin{tabular}{c|c|c|c|c|c|c|c|c|c|c|c}
    \toprule
    \multirow{3}{*}{\begin{tabular}{c}Label\\Size\\(Bytes)\end{tabular}} & \multirow{3}{*}{$n$} & \multirow{3}{*}{$m$} & \multicolumn{4}{c|}{Cong et al.~\cite{cong2021LabelledPSI}} & \multicolumn{5}{c}{\cpsi{}} \\
    & & & \multicolumn{2}{c|}{Time (s)} & \multicolumn{2}{c|}{Comm. (MB)} & $\maxbinserver$ & \multicolumn{2}{c|}{Time (s)} & \multicolumn{2}{c}{Comm. (MB)}  \\
    & & & \small{Offline} & \small{Online} & Request & Response & & Offline & Online & Request & Response \\
    \midrule
    \multirow{9}{*}{\begin{tabular}{c}No\\Label\end{tabular}}
     & \multirow{3}{*}{$2^{20}$} 
     & $1024$ & 5.6 & 0.44 & 1.3 & 1.2 & \multirow{3}{*}{296} & \multirow{3}{*}{0.11} & 7.8 & 30.0 & 0.22 \\
     & & $2048$ & 5.7 & 0.82 & 2.7 & 1.5 & &  & 7.1 & 30.0 & 0.22 \\
     & & $4096$ & 6.1 & 0.94 & 3.5 & 2.2 & &  & 7.0 & 30.0 & 0.22 \\
    \cmidrule{2-12}
    & \multirow{3}{*}{$2^{24}$} 
     & $1024$ & 99 & 1.2 & 1.8 & 2.1 & \multirow{3}{*}{3487} & \multirow{3}{*}{1.5} & 75.4 & 33.7 & 0.22 \\
     & & $2048$ & 97 & 1.5 & 3.0 & 2.3 & &  & 75.6 & 33.7 & 0.22 \\
     & & $4096$ & 97 & 1.8 & 5.2 & 2.6 & &  & 75.9 & 33.7 & 0.22 \\
    \cmidrule{2-12}
    & \multirow{3}{*}{$2^{28}$} 
     & $1024$ & 1770 & 7.3 & 3.5 & 9.5 & \multirow{3}{*}{50812} & \multirow{3}{*}{19.9} & 1144 & 39.2 & 0.22 \\
     & & $2048$ & 1720 & 7.4 & 6.0 & 9.5 & &  & 1154 & 39.2 & 0.22 \\
     & & $4096$ & 1790 & 7.7 & 8.5 & 9.8 & &  & 1141 & 39.2 & 0.22 \\
    % \bottomrule

    % \multicolumn{3}{c|}{} & \multicolumn{4}{c|}{Cong et al.~\cite{cong2021LabelledPSI}} & \multicolumn{5}{c}{\cpsi{}} \\
    %     \multicolumn{3}{c|}{} & \multicolumn{2}{c|}{Time (s)} & \multicolumn{2}{c|}{Comm. (MB)} & $\maxbinserver$ & \multicolumn{2}{c}{Time (s)} & \multicolumn{2}{c}{Comm. (MB)} \\
    %     Label & $n$ & $m$ & Offline & Online & Request & Response & & Offline & Online & Request & Response \\
    \midrule
    \midrule
        \multirow{6}{*}{32}
        & \multirow{2}{*}{$2^{20}$} & 256 & 92 & 2.4 & 4.2 & 3.8 & \multirow{2}{*}{296} & \multirow{2}{*}{0.11} & 7.2 & 30.7 & 2.9 \\
        & & 4096 & 95 & 2.4 & 4.4 & 3.5 & & & 7.0 & 30.7 & 2.9 \\
    \cmidrule{2-12}
        & \multirow{2}{*}{$2^{22}$} & 256 & 535 & 2.9 & 4.2 & 6.7 & \multirow{2}{*}{976} & \multirow{2}{*}{0.45} & 21.5 & 32.1 & 2.9 \\
        & & 4096 & 530 & 3.3 & 4.4 & 6.9 & & & 21.3 & 32.1 & 2.9 \\
    \cmidrule{2-12}
        & \multirow{2}{*}{$2^{24}$} & 256 & DNF & DNF & DNF & DNF & \multirow{2}{*}{3487} & \multirow{2}{*}{1.5} & 76.0 & 34.5 & 2.9 \\
        & & 4096 & DNF & DNF & DNF & DNF & & & 75.0 & 34.5 & 2.9 \\
    \midrule                
        \multirow{6}{*}{288}
        & \multirow{2}{*}{$2^{20}$} & 256 & 567 & 2.0 & 2.7 & 9.1 & \multirow{2}{*}{296} & \multirow{2}{*}{0.11} & 7.1 & 30.7 & 26.0 \\
        & & 4096 & 578 & 1.7 & 2.7 & 9.1 & & & 7.2 & 30.7 & 26.0 \\
    \cmidrule{2-12}
        & \multirow{2}{*}{$2^{22}$} & 256 & 3501 & 13.6 & 4.2 & 37.4 & \multirow{2}{*}{976} & \multirow{2}{*}{0.45} & 21.5 & 32.1 & 26.0 \\
        & & 4096 & 3388 & 14.2 & 4.4 & 35.3 & & & 21.5 & 32.1 & 26.0 \\
        
    \cmidrule{2-12}
        & \multirow{2}{*}{$2^{24}$} & 256 & DNF & DNF & DNF & DNF & \multirow{2}{*}{3487} & \multirow{2}{*}{1.5} & 76.5 & 34.5 & 26.0 \\
        & & 4096 & DNF & DNF & DNF & DNF & & & 76.1 & 34.5 & 26.0 \\
    \bottomrule
    \end{tabular}
    \caption{Private set intersection and Labelled PSI over large elements using the work of Cong et al. and \cpsi{}. The number of bins is set to $b=16384$ in \cpsi{} for all cases.}
    \label{tab:cong-cpsi}
\end{table*}

The communication cost of \cpsi{} is more than the work of Cong et al. in PSI and labelled PSI, but the gap narrows as the size of the labels increases.
Regarding runtime, most of the server runtime (over 95\%) in the work of Cong et al.~\cite{cong2021LabelledPSI} is spent computing a pseudo-random function for each server element in the offline phase.
They state that this step could be computed offline and stored on the server.
This is a valid assumption for some applications but does not apply to scenarios in which the server database frequently changes or when we require the overall runtime to be small.
In the case of PSI (without labels), for smaller server set sizes, the protocol of Cong et al. has a better total runtime.
However, for server sets with over $2^{24}$ elements, the total runtime of \cpsi{} is less than the work of Cong et al., with \cpsi{} requiring less time in the offline phase, but more time in the online phase.
Hence, in the case of PSI, \cpsi{} is preferable for large server sets that frequently change.

\subsection{Labelled PSI Evaluation and Comparison}

We compare with the work of Cong et al., which is the state-of-the-art in non-interactive labelled PSI.
They use the same parameters for the experiment but vary the label size.
We choose 32-byte and 288-byte labels in our experiments.
The results of this comparison are summarized in \Cref{tab:cong-cpsi}.

However, the work of Cong et al. has a high overhead when extending to labelled PSI instead of PSI.
The computation cost of Cong et al. scales with the size of the label. 
In contrast, our work requires negligible additional server time to compute labelled PSI.
The total runtime of \cpsi{} is consistently less than Cong et al. and the gap widens as the server set size increases, but \cpsi{} still requires more time in the online phase.
So, in summary, in the case of labelled PSI, \cpsi{} is consistently faster than Cong et al., particularly for large, dynamic server sets and large labels but requires more communication.

\begin{table}[t]
\centering
\resizebox{\columnwidth}{!}{
    \begin{tabular}{c|c|c|c|c|c|c|c}
    \toprule
        \multirow{3}{*}{$n$} & \multirow{3}{*}{$m$} & \multicolumn{2}{c|}{Ion et al.~\cite{ion2020deploying}} & \multicolumn{4}{c}{\cpsi{}} \\
        & & \multirow{2}{*}{\begin{tabular}{c}Time\\(s)\end{tabular}} & Comm. & $\maxbinserver$ & \multicolumn{2}{c|}{Time (s)} & Comm.\\
        & & & (MB) & & Offline & Online & (MB) \\
    \midrule
        \multirow{4}{*}{$2^{18}$}
        & 1024 & 52.9 & 20.0 & \multirow{4}{*}{100} & \multirow{4}{*}{0.01} & 3.1 & 27.6 \\
        & 2048 & 71.9 & 20.7 &  & & 2.7 & 27.6 \\
        & 4096 & 64.8 & 21.9 &  & & 2.7 & 27.6 \\
        & 8192 & 74.7 & 24.4 &  & & 2.8 & 27.6 \\
    \midrule
        \multirow{4}{*}{$2^{20}$}
        & 1024 & 199 & 78.2 & \multirow{4}{*}{296} & \multirow{4}{*}{0.11} & 7.0 & 30.0 \\
        & 2048 & 254 & 78.9 &  & & 7.1 & 30.0 \\
        & 4096 & 214 & 80.1 &  & & 7.0 & 30.0 \\
        & 8192 & 221 & 82.6 &  & & 6.8 & 30.0 \\
    \midrule
        \multirow{4}{*}{$2^{22}$}
        & 1024 & 950 & 311 & \multirow{4}{*}{976} & \multirow{4}{*}{0.44} & 21.4 & 31.4 \\
        & 2048 & 797 & 312 &  & & 21.4 & 31.4 \\
        & 4096 & 791 & 313 &  & & 21.4 & 31.4 \\
        & 8192 & 808 & 315 &  & & 21.5 & 31.4 \\
    \midrule
        \multirow{4}{*}{$2^{24}$}
        & 1024 & 3153 & 1240 & \multirow{4}{*}{3487} & \multirow{4}{*}{1.5} & 75.8 & 33.7 \\
        & 2048 & 3380 & 1240 &  & & 76.3 & 33.7 \\
        & 4096 & 3120 & 1250 &  & & 75.8 & 33.7 \\
        & 8192 & 3140 & 1250 &  & & 75.6 & 33.6 \\
    \bottomrule
    \end{tabular}
}
    \caption{PSI-Sum using \cpsi{} and the work of Ion et al.~\cite{ion2020deploying}. Times are reported in seconds, and communication is in MegaBytes. DNF denotes instances which did not finish in under one hour. The number of bins is set to $b=16384$ in \cpsi{} for all cases.}
    \label{tab:psisum-eval}
\end{table}

\subsection{Circuit PSI Evaluation and Comparison}

Our evaluation of circuit PSI is focused on non-interactive solutions, which are compatible with the unbalanced setting. Hence, we omit solutions based on 2PC~\cite{pinkas2018scalable, pinkas2019efficientcircuitbased}, which require interaction between the client and server to compute the function. 
DiPSI can extend to circuit PSI and evaluate arbitrary functions but has a prohibitively high runtime, as we saw in \Cref{tab:dipsi-cpsi}.
Recall that the work of Chen et al.~\cite{chen2017fast, chen2018LabelledPSI} and Cong et al.~\cite{cong2021LabelledPSI} cannot extend to circuit PSI.
Instead, we compare with the work of Ion et al.~\cite{ion2020deploying}, which is also a circuit PSI protocol but is limited to only a few functions: PSI-Cardinality and PSI-Sum-with-Cardinality.
We use the public implementation\footnote{\url{https://github.com/google/private-join-and-compute}} which is provided by the authors.
PSI-Stats~\cite{yingPSIStatsPrivateSet2022} performs the same operations as Ion et al. in the first two rounds of interaction and requires more operations when computing more complex functions.
In the case of computing the PSI-Sum, PSI-Stats is identical to the work of Ion et al.
Given that there is no publicly available implementation of PSI-Stats, we resort to the measurements from the work of Ion et al.

In our experiments in this section, we set the client set size to be $m\in\{1024,2048,4096,8192\}$ and vary the server set size to observe the effect on communication and computation.
The entire communication across all rounds is recorded and the runtime denotes the server runtime.
Ion et al.~\cite{ion2020deploying} map elements to 256-bit strings using hash functions, so we use the variant of \cpsi{} with large elements to have a fair comparison. \Cref{tab:psisum-eval} summarizes the results of the experiments in this section.

Note that the cardinality is leaked whilst computing the sum in the work of Ion et al.~\cite{ion2020deploying}.
In contrast, \cpsi{} does not have such leakage and is advantageous in this regard.
Moreover, Ion et al.~\cite{ion2020deploying} is very limited in the functions that it can compute.
Another observation is that the communication of the work of Ion et al. scales linearly with the server set size~\cite{ion2020deploying}, whereas the communication cost of \cpsi{} increases at a much lower rate.
For this reason, \cpsi{} is advantageous in the unbalanced setting and can scale to much larger server set sizes.
The runtime of both protocols increases as the server size increases, but \cpsi{} is over 100x times faster than the work of Ion et al.~\cite{ion2020deploying}.
One reason that \cpsi{} is faster than the work of Ion et al.~\cite{ion2020deploying} is the low per-element computation cost due to the use of batched homomorphic encryption.
In contrast, the work of Ion et al. uses expensive modular exponentiations.
The client must also perform expensive modular exponentiations, whereas in \cpsi{}, the client requires only a small amount of computation.

\subsection{Summary of Evaluation}

In our evaluation of PSI, we found that \cpsi{} demonstrates runtime that is comparable with other approaches~\cite{cong2021LabelledPSI, chen2018LabelledPSI} but requires more communication.
However, \cpsi{} excels particularly when the server size is very large and dynamic.
In the context of labelled PSI, \cpsi{} is consistently faster than existing methods but requires more communication.
The advantage of \cpsi{} increases as the label size increases.
For circuit PSI, and specifically for PSI-Sum, \cpsi{} is faster than existing non-interactive approaches~\cite{kacsmar2020differentially,ion2020deploying}, especially when dealing with larger server sets.
Additionally, \cpsi{} has a significant advantage in communication cost, as it only depends on the client set size, making it highly efficient for large server sets.
Furthermore, \cpsi{} offers the flexibility to easily extend to other functions, a capability that is not feasible with related work.

\section{Conclusion}

We propose \cpsi{}, a practical non-interactive circuit PSI protocol using homomorphic encryption.
Some state-of-the-art PSI protocols cannot extend to circuit PSI.
Those extending to circuit PSI require an interactive step with the client to compute the function or have the communication cost proportional to the server set size.
These limitations are undesirable in the unbalanced setting.
DiPSI, the only solution which does not have these limitations, has impractical runtimes.
\cpsi{} addresses all these problems and proposes an efficient circuit PSI protocol with low communication overhead.

We use multiple techniques, such as constant-weight equality operators and permutation-based hashing that greatly improve the runtime and communication cost of \cpsi{} and result in a protocol that is competitive with existing work.
We achieve competitive runtime and communication through careful optimization of parameters such as the Hamming weight.

\cpsi{} can compute the intersection of 1024 client elements with one million server elements in less than one second with less than 5 MB of communication.
Functions such as sum and cardinality can be computed with negligible additional runtime.
\cpsi{} is over four orders of magnitude faster than DiPSI, and 20x faster than the work of Ion et al.~\cite{ion2020deploying}.

\section*{Acknowledgements}
We gratefully acknowledge the support of NSERC for grants RGPIN-2023-03244, IRC-537591, the Government of Ontario and the Royal Bank of Canada for funding this research.

\bibliographystyle{plain}
\bibliography{references}

\begin{thebibliography}{10}

\bibitem{alamati2021laconic}
Navid Alamati, Pedro Branco, Nico Döttling, Sanjam Garg, Mohammad Hajiabadi, and Sihang Pu.
\newblock Laconic private set intersection and applications.
\newblock In {\em Theory of Cryptography: 19th International Conference, {TCC} 2021, Raleigh, {NC}, {USA}, November 8–11, 2021, Proceedings, Part {III}}, pages 94--125. Springer-Verlag, 2021.

\bibitem{aranha2022laconic}
Diego~F. Aranha, Chuanwei Lin, Claudio Orlandi, and Mark Simkin.
\newblock Laconic private set-intersection from pairings.
\newblock In {\em Proceedings of the 2022 {ACM} {SIGSAC} Conference on Computer and Communications Security}, {CCS} '22, pages 111--124. Association for Computing Machinery, 2022.

\bibitem{Arbitman2010BackyardCH}
Yuriy Arbitman, Moni Naor, and Gil Segev.
\newblock Backyard cuckoo hashing: Constant worst-case operations with a succinct representation.
\newblock {\em 2010 IEEE 51st Annual Symposium on Foundations of Computer Science}, pages 787--796, 2010.

\bibitem{asharov2012multiparty}
Gilad Asharov, Abhishek Jain, Adriana L{\'o}pez-Alt, Eran Tromer, Vinod Vaikuntanathan, and Daniel Wichs.
\newblock Multiparty computation with low communication, computation and interaction via threshold fhe.
\newblock In {\em Advances in Cryptology--EUROCRYPT 2012: 31st Annual International Conference on the Theory and Applications of Cryptographic Techniques, Cambridge, UK, April 15-19, 2012. Proceedings 31}, pages 483--501. Springer, 2012.

\bibitem{bourse2016fhe}
Florian Bourse, Rafaël Pino, Michele Minelli, and Hoeteck Wee.
\newblock {FHE} circuit privacy almost for free.
\newblock In {\em Proceedings, Part {II}, of the 36th Annual International Cryptology Conference on Advances in Cryptology --- {CRYPTO} 2016 - Volume 9815}, pages 62--89. Springer-Verlag, 2016.

\bibitem{certifiedsets}
Jan Camenisch and Gregory~M. Zaverucha.
\newblock Private intersection of certified sets.
\newblock In {\em 13th International Conference Financial Cryptography and Data Security ({FC})}, 2009.

\bibitem{chandran2022circuit-psi}
Nishanth Chandran, Divya Gupta, and Akash Shah.
\newblock {Circuit-PSI with Linear Complexity via Relaxed Batch OPPRF}.
\newblock In {\em 22nd Privacy Enhancing Technologies Symposium (PETS 2022)}, jun 2022.

\bibitem{HomomorphicEncryptionSecurityStandard}
Melissa Chase, Hao Chen, Jintai Ding, Shafi Goldwasser, Sergey Gorbunov, Jeffrey Hoffstein, Kristin Lauter, Satya Lokam, Dustin Moody, Travis Morrison, Amit Sahai, and Vinod Vaikuntanathan.
\newblock Security of homomorphic encryption.
\newblock Technical report, HomomorphicEncryption.org, Redmond WA, USA, July 2017.

\bibitem{chen2019multi}
Hao Chen, Ilaria Chillotti, and Yongsoo Song.
\newblock Multi-key homomorphic encryption from tfhe.
\newblock In {\em Advances in Cryptology--ASIACRYPT 2019: 25th International Conference on the Theory and Application of Cryptology and Information Security, Kobe, Japan, December 8--12, 2019, Proceedings, Part II 25}, pages 446--472. Springer, 2019.

\bibitem{chen2018LabelledPSI}
Hao Chen, Zhicong Huang, Kim Laine, and Peter Rindal.
\newblock {Labeled {PSI} from fully homomorphic encryption with malicious security}.
\newblock In {\em 24th {ACM} Conference on Computer and Communications Security ({CCS})}, 2018.

\bibitem{chen2017fast}
Hao Chen, Kim Laine, and Peter Rindal.
\newblock {Fast private set intersection from homomorphic encryption}.
\newblock In {\em 23rd {ACM} Conference on Computer and Communications Security ({CCS})}, 2017.

\bibitem{chor1997private}
Benny Chor, Niv Gilboa, and Moni Naor.
\newblock Private information retrieval by keywords.
\newblock pages 0--18, 1997.

\bibitem{ciampiCombiningPrivateSetIntersection2018}
Michele Ciampi and Claudio Orlandi.
\newblock Combining {{Private Set-Intersection}} with {{Secure Two-Party Computation}}.
\newblock In {\em Security and {{Cryptography}} for {{Networks}}: 11th {{International Conference}}, {{SCN}} 2018, {{Amalfi}}, {{Italy}}, {{September}} 5--7, 2018, {{Proceedings}}}, pages 464--482, Berlin, Heidelberg, September 2018. Springer-Verlag.

\bibitem{cong2021LabelledPSI}
Kelong Cong, Radames~Cruz Moreno, Mariana~Botelho da~Gama, Wei Dai, Ilia Iliashenko, Kim Laine, and Michael Rosenberg.
\newblock {Labeled {PSI} from Homomorphic Encryption with Reduced Computation and Communication}.
\newblock In {\em 27th {ACM} Conference on Computer and Communications Security ({CCS})}, 2021.

\bibitem{authpsi}
Emiliano~De Cristofaro, Jihye Kim, and Gene Tsudik.
\newblock Linear-complexity private set intersection protocols secure in malicious model.
\newblock In {\em 16th International Conference on the Theory and Application of Cryptology and Information Security ({ASIACRYPT})}, 2010.

\bibitem{pirpsi}
Daniel Demmler, Peter Rindal, Mike Rosulek, and Ni~Trieu.
\newblock {PIR-PSI: Scaling Private Contact Discovery}.
\newblock {\em Proceedings on Privacy Enhancing Technologies}, 2018:159--178, 10 2018.

\bibitem{duong2020catalic}
Thai Duong, Duong~Hieu Phan, and Ni~Trieu.
\newblock Catalic: Delegated psi cardinality with applications to contact tracing.
\newblock In {\em Advances in Cryptology – ASIACRYPT 2020: 26th International Conference on the Theory and Application of Cryptology and Information Security, Daejeon, South Korea, December 7–11, 2020, Proceedings, Part III}, page 870–899, Berlin, Heidelberg, 2020. Springer-Verlag.

\bibitem{fan2012somewhat}
Junfeng Fan and Frederik Vercauteren.
\newblock Somewhat practical fully homomorphic encryption.
\newblock {\em Cryptology ePrint Archive}, 2012.

\bibitem{mahdavi2022constant}
Rasoul Akhavan~Mahdavi Florian~Kerschbaum.
\newblock {Constant-weight PIR: Single-round Keyword PIR via Constant-weight Equality Operators}.
\newblock In {\em 31st USENIX Security Symposium (USENIX Security 22)}, Boston, MA, August 2022. USENIX Association.

\bibitem{freedman2016efficient}
Michael~J Freedman, Carmit Hazay, Kobbi Nissim, and Benny Pinkas.
\newblock Efficient set intersection with simulation-based security.
\newblock {\em Journal of Cryptology}, 29(1):115--155, 2016.

\bibitem{psi04}
Michael~J. Freedman, Kobbi Nissim, and Benny Pinkas.
\newblock Efficient private matching and set intersection.
\newblock In {\em 23rd International Conference on the Theory and Applications of Cryptographic Techniques ({EUROCRYPT})}, 2004.

\bibitem{ion2020deploying}
Mihaela Ion, Ben Kreuter, Ahmet~Erhan Nergiz, Sarvar Patel, Shobhit Saxena, Karn Seth, Mariana Raykova, David Shanahan, and Moti Yung.
\newblock On {{Deploying Secure Computing}}: {{Private Intersection-Sum-with-Cardinality}}.
\newblock In {\em 2020 {{IEEE European Symposium}} on {{Security}} and {{Privacy}} ({{EuroS}}\&{{P}})}, pages 370--389, Genoa, Italy, September 2020. IEEE.

\bibitem{kacsmar2020differentially}
Bailey Kacsmar, Basit Khurram, Nils Lukas, and et~al.
\newblock {Differentially Private Two-Party Set Operations}.
\newblock In {\em 2020 IEEE European Symposium on Security and Privacy (EuroS\&P)}, pages 390--404. IEEE, 2020.

\bibitem{kales2019mobile}
Daniel Kales, Christian Rechberger, Thomas Schneider, Matthias Senker, and Christian Weinert.
\newblock {Mobile Private Contact Discovery at Scale}.
\newblock In {\em 28th USENIX Security Symposium (USENIX Security 19)}, pages 1447--1464, Santa Clara, CA, August 2019. USENIX Association.

\bibitem{kirsch2010more}
Adam Kirsch, Michael Mitzenmacher, and Udi Wieder.
\newblock More robust hashing: Cuckoo hashing with a stash.
\newblock {\em SIAM Journal on Computing}, 39(4):1543--1561, 2010.

\bibitem{10.1007/11535218_15}
Lea Kissner and Dawn Song.
\newblock Privacy-preserving set operations.
\newblock In {\em Proceedings of the 25th Annual International Conference on Advances in Cryptology}, CRYPTO'05, page 241–257, Berlin, Heidelberg, 2005. Springer-Verlag.

\bibitem{lam2016breaking}
Mikkel Lambæk.
\newblock Breaking and fixing private set intersection protocols.
\newblock Cryptology ePrint Archive, Paper 2016/665, 2016.
\newblock \url{https://eprint.iacr.org/2016/665}.

\bibitem{pjac}
Tancr{\`e}de Lepoint, Sarvar Patel, Mariana Raykova, Karn Seth, and Ni~Trieu.
\newblock {Private Join and Compute from PIR with Default}.
\newblock In Mehdi Tibouchi and Huaxiong Wang, editors, {\em Advances in Cryptology -- ASIACRYPT 2021}, pages 605--634, Cham, 2021. Springer International Publishing.

\bibitem{li2019protocols}
Lucy Li, Bijeeta Pal, Junade Ali, Nick Sullivan, Rahul Chatterjee, and Thomas Ristenpart.
\newblock Protocols for checking compromised credentials.
\newblock In {\em Proceedings of the 2019 ACM SIGSAC Conference on Computer and Communications Security}, CCS '19, page 1387–1403, New York, NY, USA, 2019. Association for Computing Machinery.

\bibitem{maSecureComputationFriendlyPrivateSet2022a}
Jack P.~K. Ma and Sherman S.~M. Chow.
\newblock Secure-{{Computation-Friendly Private Set Intersection}} from {{Oblivious Compact Graph Evaluation}}.
\newblock In {\em Proceedings of the 2022 {{ACM}} on {{Asia Conference}} on {{Computer}} and {{Communications Security}}}, {{ASIA CCS}} '22, pages 1086--1097, New York, NY, USA, May 2022. Association for Computing Machinery.

\bibitem{cuckoohashing2004}
Rasmus Pagh and Flemming~Friche Rodler.
\newblock Cuckoo hashing.
\newblock {\em J. Algorithms}, 51(2):122–144, may 2004.

\bibitem{pinkas2015phasing}
Benny Pinkas, Thomas Schneider, Gil Segev, and Michael Zohner.
\newblock Phasing: Private set intersection using permutation-based hashing.
\newblock In {\em 24th USENIX Security Symposium (USENIX Security 15)}, pages 515--530, Washington, D.C., August 2015. USENIX Association.

\bibitem{pinkas2019efficientcircuitbased}
Benny Pinkas, Thomas Schneider, Oleksandr Tkachenko, and Avishay Yanai.
\newblock {Efficient Circuit-Based PSI with Linear Communication}.
\newblock In Yuval Ishai and Vincent Rijmen, editors, {\em Advances in Cryptology -- EUROCRYPT 2019}, pages 122--153, Cham, 2019. Springer International Publishing.

\bibitem{pinkas2018efficient}
Benny Pinkas, Thomas Schneider, Christian Weinert, and Udi Wieder.
\newblock Efficient circuit-based psi via cuckoo hashing.
\newblock In {\em Advances in Cryptology--EUROCRYPT 2018: 37th Annual International Conference on the Theory and Applications of Cryptographic Techniques, Tel Aviv, Israel, April 29-May 3, 2018 Proceedings, Part III 37}, pages 125--157. Springer, 2018.

\bibitem{pinkas2018scalable}
Benny Pinkas, Thomas Schneider, and Michael Zohner.
\newblock Scalable private set intersection based on ot extension.
\newblock {\em ACM Transactions on Privacy and Security (TOPS)}, 21(2):1--35, 2018.

\bibitem{reichert2021circuit}
L.~Reichert, M.~Pazelt, and B.~Scheuermann.
\newblock Circuit-based psi for covid-19 risk scoring.
\newblock In {\em 2021 IEEE International Performance, Computing, and Communications Conference (IPCCC)}, pages 1--8, Los Alamitos, CA, USA, oct 2021. IEEE Computer Society.

\bibitem{takeshita2021provably}
Jonathan Takeshita, Ryan Karl, Alamin Mohammed, Aaron Striegel, and Taeho Jung.
\newblock Provably secure contact tracing with conditional private set intersection.
\newblock In Joaquin Garcia-Alfaro, Shujun Li, Radha Poovendran, Hervé Debar, and Moti Yung, editors, {\em Security and Privacy in Communication Networks}, Lecture Notes of the Institute for Computer Sciences, Social Informatics and Telecommunications Engineering, pages 352--373. Springer International Publishing, 2021.

\bibitem{thomasProtectingAccountsCredential2019}
Kurt Thomas, Jennifer Pullman, Kevin Yeo, Ananth Raghunathan, Patrick~Gage Kelley, Luca Invernizzi, Borbala Benko, Tadek Pietraszek, Sarvar Patel, Dan Boneh, and Elie Bursztein.
\newblock Protecting accounts from credential stuffing with password breach alerting.
\newblock In {\em 28th {{USENIX Security Symposium}} ({{USENIX Security}} 19)}, pages 1556--1571, 2019.

\bibitem{Trieu2020EpioneLC}
Ni~Trieu, Kareem Shehata, P.~Saxena, R.~Shokri, and Dawn~Xiaodong Song.
\newblock Epione: Lightweight contact tracing with strong privacy.
\newblock {\em ArXiv}, abs/2004.13293, 2020.

\bibitem{viand2023verifiable}
Alexander Viand, Christian Knabenhans, and Anwar Hithnawi.
\newblock Verifiable fully homomorphic encryption, 2023.

\bibitem{TCS-070}
Udi Wieder.
\newblock Hashing, load balancing and multiple choice.
\newblock {\em Foundations and Trends{\textregistered} in Theoretical Computer Science}, 12(3--4):275--379, 2017.

\bibitem{yingPSIStatsPrivateSet2022}
Jason H.~M. Ying, Shuwei Cao, Geong~Sen Poh, Jia Xu, and Hoon~Wei Lim.
\newblock {{PSI-Stats}}: {{Private Set Intersection Protocols Supporting Secure Statistical Functions}}.
\newblock In Giuseppe Ateniese and Daniele Venturi, editors, {\em Applied {{Cryptography}} and {{Network Security}}}, pages 585--604, Cham, 2022. Springer International Publishing.

\end{thebibliography}

\appendix

\section{Proof of \Cref{lemma:fail-collide}}
\label{sec:proof-lemma}

\failcollide*
\begin{proof}
Let $\PP[B]$ denote the probability of failure due to a collision. Let $\PP[B_i]$ denote the probability of there existing a collision between any two of the elements of bin $i$, for $i\in\{1,2,\cdots,\numbins\}$. So we have 
\begin{align}
    \PP[B] \leq \sum_{i\in[\numbins]} \PP[B_i].
\end{align}

We know that, in a given bin, failure only occurs when two unequal elements, one from the server and the other from the client, have an equal constant-weight mapping.
We know that the probability of a collision between two distinct elements from the domain is $2^{-\lambda}$, where the probability is in expectation over all perfect hash functions.

Hence, combining the facts stated above, we have that 
\begin{align}
    \PP[B] \leq \sum_{i\in[\numbins]} \PP[B_i] = \sum_{i\in[\numbins]} \maxbinclient \maxbinserver 2^{-\lambda} = \numbins \maxbinclient \maxbinserver 2^{-\lambda}.
\end{align}
\end{proof}

\section{Data Preparation}
\label{sec:data-prep-code}

The algorithm for server and client data preparation is provided as Algorithm~\ref{alg:preparationappendixalgorithm}.
The client's secret key is denoted $\texttt{sk}_c$, and is used for encryption. The encryption procedure is denoted with $\enc$.
\textsc{CW-Encode} is the algorithm for mapping elements to constant-weight codewords. Precisely, we use the perfect mapping~\cite[Algorithm 3]{mahdavi2022constant} and the lossy mapping~\cite[Algorithm 8]{mahdavi2022constant} from the work of Mahdavi and Kerschbaum in the case of small and large elements, respectively. We refer the reader to the paper for the details of those algorithms.

\begin{algorithm}[t]
    \caption{Client and server data preparation}
    \label{alg:preparationappendixalgorithm}
    \begin{algorithmic}[1]
        \Procedure{ClientDataPrep}{$\clientset$}
            %hashing-to-bins
            % \InlineComment{Hashing-to-bins}
            \State Initialize $T_c$ with $b$ bins.
            \ForEach {element $x \in \mathbb{X}$}
            \State Append $x$ to the bins of $T_c$ according to the binning strategy
            \EndFor	
            \State Append dummy elements to fill each bin in $T_c$ to the max 
            
            % Encoding Elements
            % \InlineComment{Encoding each Element in the Bins}
            \ForEach{bin $k\in[\numbins]$}
                \ForEach{batch $i\in[\maxbinclient]$}
                    \State $T_c'[k][i] = \textsc{CW-Encode}(T_c[k][i],\ell,h)$
                \EndFor
            \EndFor
            
            %batching
            % \InlineComment{Client Batching and Encryption}
            \ForEach{batch $i\in[\maxbinclient]$}
                \ForEach{bit $j\in[\ell]$}
                    \State $pt[i][j] = \left[T_c'[1][i][j],T_c'[2][i][j], \cdots ,T_c'[b][i][j]\right]$
                    \State $ct_c[i][j] = \enc(pt_c[i][j], \texttt{sk}_c)$
                \EndFor
            \EndFor
            \State \textbf{return} $ct_c$
        \EndProcedure
        \vspace{3mm}
        \Procedure{ServerDataPrep}{$\serverset$}
            %hashing-to-bins
            % \InlineComment{Server Hashing-to-bins}
            \State Initialize $T_s$ with $b$ bins
            \ForEach {server element $y \in \mathbb{Y}$}
            \State Append $y$ to bins chosen by the binning strategy.
            % \State Append $y$ to $T_c[H^{(b)}(y)]$
            %     \If {$T_s[H^{(b)}(y)]$ has more than $\maxbinserver$ elements}
            %         \State Abort
            %     \EndIf
            \EndFor	
            \State Append dummy elements to fill each bin in $T_s$ to $\maxbinserver$.
    
            % Encoding Elements
            % \InlineComment{Encoding each Element in the Bins}
            \ForEach{batch $k\in[\numbins]$}
                \ForEach{batch $i\in[\maxbinserver]$}
                    \State $T_s'[k][i] = \textsc{CW-Encode}(T_c[k][i],\ell,h)$
                \EndFor
            \EndFor
            
            %batching
            % \InlineComment{Server Batching}
            \ForEach{batch $i\in[\maxbinserver]$}
                \ForEach{bit $j\in[\ell]$}
                    \State $pt_s[i][j] = \left[T_s'[1][i][j],T_s'[2][i][j], \cdots ,T_s'[b][i][j]\right]$
                \EndFor
            \EndFor	
            \State \textbf{return} $pt_s$
        \EndProcedure
    \end{algorithmic}
\end{algorithm}

\section{Large Labels}
\label{sec:labelled-large}

\Cref{alg:server-comp-labelled} can be extended to the case where server labels are larger than one plaintext slot.
The expensive step of the algorithm, which is matching client and server elements, does not need to be repeated. Instead, we simply repeat line 4 of \Cref{alg:server-comp-labelled} for as many times that is required, depending on the size of the label.

% \begin{algorithm}
%     \caption{\textsc{ServerComp\textsubscript{PSI-Labelled-Large}}}\label{alg:server-comp-labelled-large}
%     \label{alg:algorithmservercompappendix}
%     \begin{algorithmic}[1]
%         \ForEach {client batch $i\in[\maxbinclient]$}
%             % \InlineComment{Intersection of the $i$th client batch and $T'_s$}	
%             \ForEach{server batch $i'\in[\maxbinserver]$}
%                 \State $ct_{eq}[i][i'] \leftarrow \texttt{Arith-CW-Eq}(ct_c[i], pt_s[i'])$
%             \EndFor
%             \State $ct_{res}[i] \leftarrow \sum_{i'} val_s[i'] \cdot ct_{eq}[i][i']$
%         \EndFor
%         \State Send $ct_{res}$ to the client
%     \end{algorithmic}
% \end{algorithm}

\section{PSI-kth-Match.}
Another function that the server can compute over the intersection is returning only one element in the intersection.
This can be used to iteratively return the results of the protocol.
For example, a data scientist may wish to see a sample of items in the intersection before deciding whether to recieve the entire intersection.

For simplicity, we only describe the cleartext algorithm for finding the k$^{th}$ match in a vector.
The algorithm is designed such that it can be efficiently computed using homomorphic encryption.
We include correctness and security proof to show that this algorithm can compute the k$^{th}$ match whilst not revealing any other information about the intersection.
The full details on how to implement this algorithm using homomorphic encryption and the FV library are left for future work.

Assume that $I\in\{0,1\}^{n}$ denotes a vector with ones in some indices. The objective is to compute $M$ such that $M[i_k]=0$ where $i_k$ is the position of the k$^{th}$ one in $I$. \Cref{alg:server-comp-kth-match} shows the algorithm that can achieve this.

\begin{algorithm}
    \caption{Algorithm for k$^{th}$ match}\label{alg:server-comp-kth-match}
    \begin{algorithmic}[1]
        \Procedure{ComputePSIkthMatch}{$I\in \{0,1\}^{n}, k$}
            \ForEach{$i \in [n]$}
                \State $r \xleftarrow{\$} \ZZ_p$
                \State $M[i] \leftarrow r \cdot (k \cdot I[i] - 1 - \sum_{i'<i} I[i'])$
            \EndFor
            \State \textbf{return} $M$
        \EndProcedure
    \end{algorithmic}
\end{algorithm}

\begin{theorem}
    \Cref{alg:server-comp-kth-match} is correct, i.e., $M$ is zero at the index of the k$^{th}$ one in $I$, and non-zero in all the indicies.
\end{theorem}

\begin{proof}
Let $i_1 < i_2 < \cdots < i_s$ be the indicies such that $I[i_1] = I[i_2] = \cdots = I[i_s] = 1$ and $I[i]=0$ for $i\neq i_j$. If $M[i] = k \cdot I[i] - 1 - \sum_{i'\leq i} I[i']$ then

$$
    M[i_k] = k \cdot I[i_k] - 1 - \sum_{i' < i_k} I[i'] = k - 1 - (k-1) = 0
$$

Moreover, 
$$
    j\neq k \Rightarrow M[i_j] = k \cdot I[i_j] - 1 - \sum_{i' < i_j} I[i'] = k - j \neq 0
$$
$$
    i_{j} < i < i_{j+1} \Rightarrow M[i] = k \cdot I[i] - 1 - \sum_{i'\leq i} I[i'] = -1 - j \neq 0
$$
which proves the theorem.
\end{proof}

The security of this algorithm is proven by showing that $M$ reveals nothing about $I$ other than the position of the k$^{th}$ one in the array.

\begin{theorem}
    If $M$ is defined as in \Cref{alg:server-comp-kth-match}, then $M$ reveals nothing about $I$ other than the index of the k$^{th}$ one.
\end{theorem}

\begin{proof}
As shown in the correctness proof,
$$
    k \cdot I[i] - 1 - \sum_{i'\leq i} I[y'] \neq 0
$$
for $i\neq i_k$. Moreover, due to the multiplication of $r$, which is uniformly random, $M[i]$ is also uniformly random, for $i\neq i_k$. Hence, $M[i']$ for $i'\neq i_k$ does not reveal any information about $I$.
\end{proof}

\section{Optimization for different values of $\numbins$}
\Cref{fig:ell-per-ex-comp-4096} showed the relationship between runtime and the Hamming weight for a fixed range of $b$. While there is not a close formula to show the effect of $b$, we can see in the graph below that for a larger $b$, the effect is roughly the same.

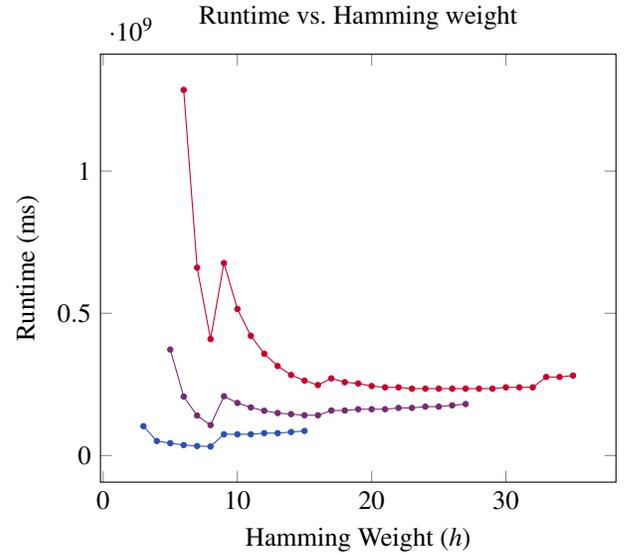
\begin{figure}[H]
    \centering
    % \subfloat{    
    \begin{tikzpicture}
        \begin{axis}[
            table/col sep=comma,
            xlabel={Hamming Weight ($\hw$)},
            ylabel={Runtime (ms)},
            title={Runtime vs. Hamming weight}
        ]
            \addplot [mark=*, mark size=1pt, mark options={color=pepsiblue}, color=pepsiblue] table [y=CodeLengthWeight, x=HammingWeight]{data/output-bt=16-b=8192.csv};
            \addplot [mark=*, mark size=1pt, mark options={color=pepsipurple}, color=pepsipurple] table [y=CodeLengthWeight, x=HammingWeight]{data/output-bt=32-b=8192.csv};
            \addplot [mark=*, mark size=1pt, mark options={color=pepsired}, color=pepsired] table [y=CodeLengthWeight, x=HammingWeight]{data/output-bt=48-b=8192.csv};
        \end{axis}
    \end{tikzpicture}
    % }~
    % \subfloat{    
    % \begin{tikzpicture}[scale=0.45]
    %     \begin{axis}[
    %         table/col sep=comma,
    %         xlabel={\Large Hamming Weight ($\hw$)},
    %         ylabel={\Large $\ell(\effectivebitlength,\hw) \cdot N q$},
    %         title={\Large $\effectivebitlength=32$}
    %     ]
    %         \addplot [mark=*, mark size=2pt, mark options={color=pepsiblue}, color=pepsiblue] table [y=CodeLengthWeight, x=HammingWeight]{data/output-bt=32-b=8192.csv};
    %     \end{axis}
    % \end{tikzpicture}
    % }
    
    % \subfloat{    
    % \begin{tikzpicture}[scale=0.45]
    %     \begin{axis}[
    %         table/col sep=comma,
    %         xlabel={\Large Hamming Weight},
    %         ylabel={\Large $\ell(\effectivebitlength,h) \cdot N q$},
    %         title={\Large $\effectivebitlength=48$}
    %     ]
    %         \addplot [mark=*, mark size=2pt, mark options={color=pepsiblue}, color=pepsiblue] table [y=CodeLengthWeight, x=HammingWeight]{data/output-bt=48-b=8192.csv};
    %     \end{axis}
    % \end{tikzpicture}
    % }
    % \subfloat{    
    % \begin{tikzpicture}[scale=0.45]
    %     \begin{axis}[
    %         table/col sep=comma,
    %         xlabel={\Large Hamming Weight},
    %         ylabel={\Large $\ell(\effectivebitlength,h) \cdot N q$},
    %         title={\Large $\effectivebitlength=64$}
    %     ]
    %         \addplot [mark=*, mark size=2pt, mark options={color=pepsiblue}, color=pepsiblue] table [y=CodeLengthWeight, x=HammingWeight]{data/output-bt=64-b=8192.csv};
    %     \end{axis}
    % \end{tikzpicture}
    % }
    \caption{Code length as a function of the Hamming weight for $\effectivebitlength\in\{16, 32, 48, 64\}$ for $4096<\numbins\leq 8192$. The minimum occurs for a Hamming weight of 8, 8, and 23, respectively.}
    \label{fig:ell-per-ex-8192}
\end{figure}

\end{document}